\documentclass[11pt]{article}
\usepackage{ifthen}
\usepackage{mdwlist}
\usepackage{amsmath,amssymb,amsfonts,amsthm}
\usepackage{bm}
\usepackage{bbm}
\usepackage[colorlinks=true,linkcolor=blue,citecolor=blue,urlcolor=blue]{hyperref}
\usepackage{enumitem}
\usepackage{graphicx}
\usepackage{xspace}
\usepackage{verbatim}
\usepackage{algorithm}
\usepackage{algpseudocode}
\usepackage[margin=1in]{geometry}
\usepackage{color}
\usepackage{thm-restate}
\usepackage{latexsym}
\usepackage{epsfig}
\usepackage[symbol]{footmisc}
\usepackage[colorinlistoftodos,textsize=scriptsize]{todonotes}
\def\withcolors{1}
\def\withnotes{1}
\def\eps{\ve}
\renewcommand{\epsilon}{\ve}
\def\ve{\varepsilon}

\newcommand{\E}{\mbox{\bf E}}
\newcommand{\Var}{\mbox{\bf Var}}

\newcommand{\pr}[2][]{\mathrm{Pr}\ifthenelse{\not\equal{}{#1}}{_{#1}}{}\!\left[#2\right]}
\newcommand{\R}{\mathbb{R}}

\newcommand{\dtv}{d_{\mathrm {TV}}}

\newcommand{\id}{\mathbb{I}}

\def \Paren#1{{\left({#1}\right)}}

\newtheorem{theorem}{Theorem}

\newtheorem{remark}[theorem]{Remark}
\newtheorem{fact}[theorem]{Fact}
\newtheorem{lemma}[theorem]{Lemma}
\newtheorem{claim}[theorem]{Claim}
\newtheorem{corollary}[theorem]{Corollary}

\newtheorem{definition}[theorem]{Definition}
\newtheorem{question}[theorem]{Question}

\numberwithin{theorem}{section} 
\numberwithin{nontheorem}{section} 
\numberwithin{proposition}{section} 
\numberwithin{observation}{section} 
\numberwithin{remark}{section} 
\numberwithin{fact}{section} 
\numberwithin{lemma}{section} 
\numberwithin{claim}{section} 
\numberwithin{corollary}{section} 
\numberwithin{case}{section} 
\numberwithin{dfn}{section} 
\numberwithin{definition}{section} 
\numberwithin{question}{section} 
\numberwithin{openquestion}{section} 
\numberwithin{res}{section}

\ifnum\withcolors=1
  \newcommand{\gcolor}[1]{{\color{red}#1}} 
  
  \newcommand{\hcolor}[1]{{\color{blue}#1}} 
\else

  \newcommand{\gcolor}[1]{{#1}}
  
  \newcommand{\hcolor}[1]{{#1}} 
\fi

\ifnum\withnotes=1
  \newcommand{\gnote}[1]{\par\gcolor{\textbf{G: }\sf #1}} 
  \newcommand{\gfootnote}[1]{\footnote{{\bf \gcolor{Gautam}}: {#1}}}
  
  \newcommand{\hnote}[1]{ \hcolor{\textbf{H: }\sf #1}} 
  \newcommand{\snote}[1]{ \hcolor{\textbf{SW: }\sf #1}} 

\else
  \newcommand{\gnote}[1]{}
  \newcommand{\hnote}[1]{}
    \newcommand{\snote}[1]{}

  \newcommand{\gfootnote}[1]{}

\fi
\newcommand{\ignore}[1]{\leavevmode\unskip} 

\newcommand{\ellone}[1]{||#1||_{\ell_1}}
\newcommand{\polylog}{\mathrm{polylog}}
\newcommand{\Dinfty}{D_\infty}
\newcommand{\Lap}{\mathrm{Lap}}
\DeclareMathOperator*{\argmin}{argmin}
\newcommand{\DKL}{D_{\mathrm{KL}}}

\usepackage{mathtools}

\title{Locally Private Hypothesis Selection}

\author {
Sivakanth Gopi\thanks{Microsoft Research. {\tt sigopi@microsoft.com}.}
\and
Gautam Kamath\thanks{Cheriton School of Computer Science, University of Waterloo. {\tt g@csail.mit.edu}. Supported by a University of Waterloo startup grant. Part of this work was done while supported as a Microsoft Research Fellow, as part of the Simons-Berkeley Research Fellowship program at the Simons Institute for the Theory of Computing, and while visiting Microsoft Research Redmond.}
\and
Janardhan Kulkarni\thanks{Microsoft Research. {\tt jakul@microsoft.com}.}
\and
Aleksandar Nikolov\thanks{Department of Computer Science, University of Toronto. {\tt anikolov@cs.toronto.edu}. Supported by Ontario ERA, and NSERC Discovery grants. Part of this work was done while visiting the Simons Institute for the Theory of Computing.}
\and
Zhiwei Steven Wu\thanks{University of Minnesota. {\tt zsw@umn.edu}. Supported in part by the NSF FAI Award \#1939606, a Google Faculty Research Award, a J.P. Morgan Faculty Award, a Facebook Research Award, and a Mozilla Research Grant. Part of this work was done while visiting the Simons Institute for the Theory of Computing.}
\and
Huanyu Zhang\thanks{Cornell University. {\tt hz388@cornell.edu}. Supported by NSF \#1815893 and by NSF \#1704443. This work was partially done while the author was an intern at Microsoft Research Redmond.}
}

\begin{document}
\maketitle \footnotetext{Authors are in alphabetical order.}

\begin{abstract}
  We initiate the study of hypothesis selection under local differential privacy.
  Given samples from an unknown probability distribution $p$ and a set  $\mathcal{Q}$ of $k$ probability distributions, we aim to output, under the constraints of $\varepsilon$-local differential privacy, a distribution from $\mathcal{Q}$ whose total variation distance to $p$ is comparable to the best such distribution.
  This is a generalization of the classic problem of $k$-wise simple hypothesis testing, which corresponds to when $p \in \mathcal{Q}$, and we wish to identify $p$.
  Absent privacy constraints, this problem requires $O(\log k)$ samples from $p$, and it was recently shown that the same complexity is achievable under (central) differential privacy.
  However, the naive approach to this problem under local differential privacy would require $\tilde O(k^2)$ samples.

  We first show that the constraint of local differential privacy incurs an exponential increase in cost: any algorithm for this problem requires at least $\Omega(k)$ samples.
  Second, for the special case of $k$-wise simple hypothesis testing, we provide a non-interactive algorithm which nearly matches this bound, requiring $\tilde O(k)$ samples.
  Finally, we provide sequentially interactive algorithms for the general case, requiring $\tilde O(k)$ samples and only $O(\log \log k)$ rounds of interactivity.
  Our algorithms for the general case are achieved through a reduction to maximum selection with adversarial comparators, a problem of independent interest for which we initiate study in the parallel setting.
  For this problem, we provide a family of algorithms for each number of allowed rounds of interaction $t$, as well as lower bounds showing that they are near-optimal for every $t$.
  Notably, our algorithms result in exponential improvements on the round complexity of previous methods.
\end{abstract}

\section{Introduction}
Perhaps the most fundamental question in statistics is that of simple hypothesis testing.
Given two known distributions $p$ and $q$, and a dataset generated according to one of these distributions, the goal is to determine which distribution the data came from.
The optimal solution to this problem is the likelihood-ratio test, as shown by Neyman and Pearson~\cite{NeymanP33}.
This problem can be generalized in two ways that we consider in this paper.
First, rather than just two distributions, one can consider a setting where the goal is to select from a set of $k$ distributions.
We refer to this setting as \emph{$k$-wise simple hypothesis testing}.
Furthermore, the data may not have been generated according to \emph{any} distribution from the set of known distributions -- instead, the goal is to just select a distribution from the set which is competitive with the best possible (in an appropriate distance measure).
This problem is the core object of our study, and we denote it as \emph{hypothesis selection}.

The hypothesis selection problem appears naturally in a number of settings.
For instance, we may have a collection of distribution learning algorithms that are effective under different assumptions on the data, but it is unknown which ones hold in advance.
Hypothesis selection allows us to simply run all of these algorithms in parallel and pick a good output from these candidate distributions afterwards.
More generally, a learning algorithm may first ``guess'' various parameters of the unknown distribution during and for each guess produce a candidate output distribution. Hypothesis selection allows us to pick a final result from this set of candidates.
Finally, near-optimal sample complexity bounds can often be derived by enumerating all possibilities within some parametric class of distributions (i.e.,  a cover) and then applying hypothesis selection with this enumeration as the set of hypotheses~\cite{DevroyeL01}.

Classical work (e.g.,~\cite{Yatracos85, DevroyeL96, DevroyeL97, DevroyeL01}) on these problems has shown that, even in the most general setting of hypothesis selection, there are effective algorithms with sample complexity scaling only \emph{logarithmically} in the number of candidate hypotheses.
Building on this, there has been significant study into hypothesis selection with additional desiderata, including computational efficiency, robustness, weaker access to hypotheses, and more (e.g.,~\cite{MahalanabisS08, DaskalakisDS12b, DaskalakisK14, SureshOAJ14, AcharyaJOS14b, DiakonikolasKKLMS16, AcharyaFJOS18, BousquetKM19, BunKSW19}).

One consideration which has not received significant attention in this setting is that of \emph{data privacy}.
The dataset may be comprised of personally sensitive data, including medical records, location history, or salary information, and classical hypothesis selection algorithms may violate the privacy of individuals who provided the data. 
Motivated by this issue, our goal is to perform our statistical analysis while ensuring that the output does not reveal significant information about any individual datapoint.
We will be concerned with the formalization of this principle as \emph{differential privacy}~\cite{DworkMNS06}, which can be seen as the gold standard for modern data privacy.

We first distinguish between two common definitions of differential privacy.
The first is \emph{central differential privacy} (also known as the \emph{trusted curator} setting)~\cite{DworkMNS06}, in which users transmit their data to a central server without any obfuscation, and the algorithm operates on this dataset with the restriction that its final output must be appropriately privatized.
The second is \emph{local differential privacy} (LDP)~\cite{Warner65, EvfimievskiGS03, KasiviswanathanLNRS11}, in which users trust no one: each individual privatizes their own data before sending it to the central server.
In some sense, LDP places the privacy barrier closer to the users, and as a result, has seen adoption in practice by a number of companies that analyze sensitive user data, including Google~\cite{ErlingssonPK14}, Microsoft~\cite{DingKY17}, and Apple~\cite{AppleDP17}.

Recently, Bun, Kamath, Steinke, and Wu~\cite{BunKSW19} showed that under the constraint of central differential privacy, one can still perform hypothesis selection with sample complexity which scales logarithmically in the number of hypotheses.
A priori, it was not clear that this would be possible.
Non-privately, one can apply methods which essentially ask ``Which of these two distributions fits the data better?'' for all $O(k^2)$ pairs of hypotheses.
Crucially, one can reuse the same set of $O(\log k)$ samples for all such comparisons (rather than drawing fresh samples for each one), and accuracy can be proved by a Chernoff and union bound style argument.
A naive privatization of this method would result in a polynomial dependence on $k$, due to issues arising from sample reuse and the composition of privacy losses.
\cite{BunKSW19} avoid this issue by a careful application of tools from the differential privacy literature (i.e., the exponential mechanism~\cite{McSherryT07}), achieving an $O(\log k)$ sample complexity.
However, their method relies upon techniques which are not available in the local model of differential privacy.
Indeed, at first glance, it may not be clear how to improve upon an $\tilde O(k^2)$ sample complexity in the local model, achieved by simply using a fresh set of samples for each comparison, and using randomized response to privately perform the comparison.
This raises the question: what is the sample complexity of hypothesis selection under local differential privacy?
Can the problem be solved with a logarithmic dependence of the number of samples on the number of candidate hypotheses?
Or do we require a polynomial number of samples?

\subsection{Results, Techniques, and Discussion}
To describe our results, we more formally define the problems of $k$-wise simple hypothesis testing and hypothesis selection. For a definition of the total variation distance $\dtv(p,q)$, we refer the reader to Section~\ref{sec:pre}.

\begin{definition}
  Suppose we are given a set of $n$ data points $X_1, \dots, X_n$, which are sampled i.i.d.\ from some (unknown) distribution $p$, and a set of $k$ distributions $\mathcal{Q} = \{q_1, \dots, q_k\}$.
  The goal is to output a distribution $\hat q \in \mathcal{Q}$ such
  that $\dtv(p,\hat q) \leq c \min_{q^* \in \mathcal{Q}} \dtv(p, q^*)
  + \alpha$, for some $c =c(\alpha, k)$.
  
  We refer to the value of $c(\alpha,k)$ as the \emph{agnostic approximation factor}.
  If $c(\alpha,k)$ is an absolute constant, then we denote this problem as \emph{hypothesis selection}.
  If $c(\alpha, k)$ grows with $k$ and $\frac1\alpha$, we refer to
  this problem as \emph{weak hypothesis selection}. If we require that
  $p \in \mathcal{Q}$, that $\min_{i\neq j} \dtv(q_i, q_j) \ge \alpha$, and that the algorithm must correctly identify $p$, then we denote this problem as \emph{$k$-wise simple hypothesis testing}.
\end{definition}

We assume the reader is familiar with the notion of $\ve$-local differential privacy ($\ve$-LDP); a formal definition appears in Section~\ref{sec:pre}.

Our first result shows that $k$-wise simple hypothesis testing (and thus, hypothesis selection) requires $\Omega(k)$ samples.
\begin{restatable}{theorem}{lb}\label{thm:fully-interactive-lb}
    Let $\ve \in (0,1)$.
  Suppose $M$ is an $\ve$-LDP protocol that solves the $k$-wise
  simple hypothesis testing problem  with probability at least $1/3$
  when given $n$ samples from some distribution $p \in \mathcal{Q}$,
  for any set $\mathcal{Q} = \{q_1, \ldots, q_k\}$ such
  that $\min_{i \neq j} \dtv(q_i, q_j) \ge \alpha$. Then
  $n = \Omega\left( \frac{k}{ \alpha^2 \ve^2}\right).$
\end{restatable}
The theorem above shows that the cost of hypothesis testing is exponentially larger under local differential privacy than under central differential privacy (i.e., $\Omega(k)$ versus $O(\log k)$), and it holds even when the LDP protocol is allowed the power of full interactivity.
The construction used to prove this lower bound is the problem of $1$-sparse mean estimation, previously identified as a problem of interest by Duchi, Jordan, and Wainwright~\cite{DuchiJW13, DuchiJW17}. The lower bound follows from results in \cite{DuchiR19}.
Given the construction, our result can be seen as a translation of existing results. Further details are given in Section~\ref{sec:lb}.

With a lower bound of $\Omega(k)$ samples, and the aforementioned naive upper bound of $\tilde O(k^2)$ samples, the problem remains to identify the correct sample complexity.
We provide two different algorithms which require $\tilde O(k)$ samples, nearly matching this lower bound.
The first is for the special case of $k$-wise simple hypothesis testing, and is a non-interactive protocol -- each user only sends a message to the curator once, independently of the messages sent by other users. 
The second solves the more general problem of hypothesis selection, but requires sequential interactivity (albeit only $O(\log \log k)$ rounds of interaction): users still only send a message to the curator once, but the curator may request different types of messages from later users based on the messages sent by earlier users.
Less interaction in a protocol is generally preferred, and the role and power of interactivity in local differential privacy is one of the most significant questions in the area (see, e.g.~\cite{KasiviswanathanLNRS11,JosephMNR19,DanielyF19,DuchiR19,JosephMR20}).

Our first algorithmic result gives a non-interactive mechanism with $\tilde O(k)$ sample complexity for sufficiently well separated instances. Define $\beta :=  \min_{q \in \mathcal{Q}} \dtv(p, q)$.

\begin{theorem}
  \label{thm:1Round}
  For every $\varepsilon \in [0,1)$, there is a non-interactive $\epsilon$-LDP algorithm that with
  probability at least $1-1/k^2$ outputs a distribution
  $\hat q \in \mathcal{Q}$ such that $\dtv(p,\hat q) \leq \alpha$, if the number of samples
  $n \gg k (\log k)^3/(\alpha^4\epsilon^2)$ and $\beta \ll \alpha^2/\log k$.\footnote{We use $A\ll B$ to denote that $A\le c B$ for some sufficiently small constant $c>0$. Similarly we use $A \gg B$ to denote that $A \ge C B$ for some sufficiently large constant $C>0.$ $A \lesssim B$ is used interchangeably with $A=O(B)$. Similarly $A\gtrsim B$ is used interchangeably with $A=\Omega(B).$}
\end{theorem}

We prove the theorem in Section \ref{sec:non-adaptive}. While somewhat more general, the above theorem immediately gives a non-interactive $\tilde O(k)$-sample algorithm for the important special case of LDP $k$-wise simple hypothesis testing.

\begin{corollary}
  \label{cor:1rRealizable}
Suppose our instance of hypothesis testing is such that $p \in \mathcal{Q}$ and all distributions in $\mathcal{Q}$ are $\Omega(\alpha)$-far from each other in total variation distance.
For $\varepsilon \in [0,1)$, there exists a non-interactive $\varepsilon$-LDP algorithm which identifies $p$ with high probability, given $n = O\left(\frac{k \log^3 k}{\alpha^4 \varepsilon^2}\right)$ samples.
\end{corollary}

Our algorithm is based on a noised log-likelihood test, though significant massaging and manipulation of the problem instance is required to achieve an acceptable sample complexity. 
In our algorithm, the users are divided into $k$ groups. Each user in the $i^{th}$ group sends the log-likelihood (with some Laplace noise added for privacy) of observing the sample given to the user if the true distribution was $q_i$. The log-likelihoods from all the users in the $i^{th}$ group are aggregated and the most likely distribution is output. Alternatively, we can also think of our algorithm as using the samples from the $i^{th}$ group to estimate KL-divergences between the unknown distribution and $q_i$ and finally outputting the closest distribution.
For this approach to work, we need all the log-likelihoods to be
bounded. We achieve this by a {\em flattening lemma} which makes all
the distributions close to uniform, while preserving their total
variation distances. Moreover, this flattening can be implemented
locally by the users transforming their samples from the original
distribution. We believe that our flattening lemma may have
applications in other DP problems.

\smallskip
Our second algorithmic result is a $O(\log \log k)$-round sequentially interactive $\tilde O(k)$-sample algorithm for LDP hypothesis selection.

\begin{corollary}[Informal version of Corollary~\ref{cor:set-params-ldp}]
  \label{cor:informal-sequential}
  Suppose we are given $n$ samples from an unknown distribution $p$ and a set of descriptions of $k$ distributions $\mathcal{Q}$.
  There exists an algorithm which identifies a distribution $\hat q \in \mathcal{Q}$, such that $\dtv(p,\hat q) \leq 27  \min_{q^* \in \mathcal{Q}} \dtv(p,q^*) + O(\alpha)$ with probability $9/10$.
  The algorithm is $\ve$-LDP, requires $O(\log \log k)$ rounds of sequential interactivity, and $n = O\left(\frac{k \log k \log \log k}{\alpha^2 \varepsilon^2}\right)$ samples.
\end{corollary}

The $k$-wise simple hypothesis testing and hypothesis selection
problems can also be studied in the Statistical Queries (SQ) model of~\cite{Kearns98}. In this model, rather than being given samples
from a distribution $p$, the algorithm can ask queries specified by
bounded functions $\phi$, and get a (possibly adversarial) additive
$\tau$-approximation to the expectation of $\phi$ under $p$, where the
parameter $\tau$ is usually called the tolerance. For distributional
problems, \cite{KasiviswanathanLNRS11} showed that sample complexity in the LDP model is equivalent up to
polynomial factors to complexity in the SQ model, measured in terms of
the number of queries and the inverse tolerance
$\frac1\tau$. In particular, this
connection and our lower bound in
Theorem~\ref{thm:fully-interactive-lb} imply that $k$-wise simple
hypothesis testing in the SQ model requires that either the number of
queries or $\frac1\tau$ be polynomial in $k$. Because of the
polynomial loss, however, our precise study of the sample complexity
of these problems does not immediately translate to the SQ model. We
remark that both the 1-round algorithm in
Corollary~\ref{cor:1rRealizable}, and the algorithm in
Corollary~\ref{cor:informal-sequential} can be implemented in the SQ
model, and require, respectively, $1$ round and $O(\log \log k)$ rounds
of adaptive queries. Understanding the precise relationship between
the number of queries, the tolerance parameter, and the number of
rounds of adaptivity for solving hypothesis selection in the SQ model
is an interesting direction for future work.

Interestingly, Corollary~\ref{cor:informal-sequential} is derived as a consequence of a connection to maximum selection with adversarial comparators, a problem of independent interest.
This connection was previously established in works by Acharya, Falahatgar, Jafarpour, Orlitsky, and Suresh~\cite{AcharyaJOS14b, AcharyaFJOS18}. 
Prior work, however, has not exploited this connection under LDP constraints. Given the aforementioned importance of interactivity in the LDP setting, we initiate a study of the maximum selection with adversarial comparators problem from the perspective of understanding the trade-off between the number of rounds of parallel comparisons, and the total number of comparisons. 
The problem is as follows: we are given a set of items of unknown value, and we can perform comparisons between pairs of items.
If the value of the items is significantly different, the comparison will correctly report the item with the larger value.
If the values are similar, then the result of the comparison may be arbitrary.
The goal is to output an item with value close to the maximum.
We wish to minimize the total number of comparisons performed, as well as the number of rounds of interactivity.

Our main result for this setting gives a family of algorithms and lower bounds, parameterized by the number of rounds used (denoted by $t$).
Setting $t = O(\log \log k)$ yields Corollary~\ref{cor:informal-sequential}.
\begin{theorem}[Restatement of Theorems~\ref{thm:better-t-round} and~\ref{thm:lower-t-round}]
  \label{thm:informal-comparisons}
  For every $t \in \mathbb{Z}^+$, there exists a $t$-round protocol which, with probability $9/10$, approximately solves the problem of parallel approximate maximum selection with adversarial comparators from a set of $k$ items.
  The algorithm requires $O(k^{1 + \frac{1}{2^t-1}}t)$ comparison queries.
  Furthermore, any algorithm which provides these guarantees requires $\Omega(\frac{ k^{1 + \frac{1}{2^t-1}}}{3^t})$ comparison queries.
\end{theorem}
For each number of rounds $t$, we prove an upper bound and an almost-matching lower bound. 
In order to get down to a near-linear number of comparisons, we require $O(\log \log k)$ rounds, which is exponentially better than the $O(\log k)$ rounds required by previous algorithms.
Interestingly, in this setting, while maximum selection (with standard comparisons) with $\tilde O(k)$ queries is achievable in only 3 rounds, we show that $\Theta(\log \log k)$ rounds are both necessary and sufficient to achieve a near-linear number of comparisons when the results might be adversarial.

Our upper bounds follow by carefully applying a recursive tournament structure: in each round, we partition the input into appropriately-sized smaller groups, perform all pairwise-comparisons within each group, and send only the winners to the next round.
Additional work is needed to prevent the quality of approximation from decaying as the number of rounds increases. For the lower bound, we restate the problem as a game, in which the adversary constructs a random complete directed graph with a unique sink, and the algorithm queries the directions of edges, and tries to identify the sink in the smallest number of queries and rounds. We give a strategy in which the adversary constructs a layered graph with $t+1$ layers, where $t$ is the number of rounds in the game. We can guarantee that, if the algorithm does not make enough queries, then even after conditioning on the answers to the queries in the first $q$ rounds, the last $t+1 - q$ layers of the graph remain sufficiently random, so that the algorithm cannot guess the sink with reasonable probability. In particular, after $t$ rounds, there is still enough randomness in the $(t+1)$-st layer to make sure that algorithm cannot guess the sink correctly with high probability.

A self-contained description of the connection between hypothesis selection and maximum selection with adversarial comparators, as well as our upper and lower bounds, appear in Section~\ref{sec:comparators}.

\subsection{Related Work}
\label{sec:related}
As mentioned before, our work builds on a long line of investigation on hypothesis selection.
This style of approach was pioneered by Yatracos~\cite{Yatracos85}, and refined in subsequent work by Devroye and Lugosi~\cite{DevroyeL96, DevroyeL97, DevroyeL01}.
After this, additional considerations have been taken into account, such as computation, approximation factor, robustness, and more~\cite{MahalanabisS08, DaskalakisDS12b, DaskalakisK14, SureshOAJ14, AcharyaJOS14b, DiakonikolasKKLMS16, AcharyaFJOS18, BousquetKM19, BunKSW19}.
Most relevant is the recent work of Bun, Kamath, Steinke, and Wu~\cite{BunKSW19}, which studies hypothesis selection under central differential privacy.
Our results are for the stronger constraint of local differential privacy.

Versions of our problem have been studied under both central and local differential privacy.
In the local model, the most pertinent result is that of Duchi, Jordan, and Wainwright~\cite{DuchiJW13, DuchiJW17}, showing a lower bound on the sample complexity for simple hypothesis testing between two known distributions. 
This matches folklore upper bounds for the same problem.
However, the straightforward way of extending said protocol to $k$-wise simple hypothesis testing would incur a cost of $\tilde O(k^2)$ samples.
Other works on hypothesis testing under local privacy include~\cite{GaboardiR18,Sheffet18,AcharyaCFT19,AcharyaCT19,JosephMNR19}.
In the central model, some of the early work was done by the Statistics community~\cite{VuS09,UhlerSF13}.
More recent work can roughly be divided into two lines -- one attempts to provide private analogues of classical statistical tests~\cite{WangLK15,GaboardiLRV16,KiferR17,KakizakiSF17,CampbellBRG18,SwanbergGGRGB19,CouchKSBG19}, while the other focuses more on achieving minimax sample complexities for testing problems~\cite{CaiDK17, AcharyaSZ18, AliakbarpourDR18, AcharyaKSZ18, CanonneKMUZ19, AliakbarpourDKR19, AminJM19}. 
While most of these focus on composite hypothesis testing, we highlight~\cite{CanonneKMSU19} which studies simple hypothesis testing.
Work of Awan and Slavkovic~\cite{AwanS18} gives a universally optimal test for binomial data, however Brenner and Nissim~\cite{BrennerN14} give an impossibility result for distributions with domain larger than $2$.
For further coverage of differentially private statistics, see~\cite{KamathU20}.

We are the first to study parallel maximum selection with adversarial comparators.
Prior work has investigated (non-parallel) maximum selection and sorting with adversarial comparators~\cite{AjtaiFHN09, AcharyaJOS14b, AcharyaFJOS18}.
Works by Acharya, Falahatgar, Jafarpour, Orlitsky, and Suresh established the connection with hypothesis selection~\cite{AcharyaJOS14b, AcharyaFJOS18}.
The parallelism model we study here was introduced by Valiant~\cite{Valiant75}, for parallel comparison-based problems with non-adversarial comparators.
This has inspired a significant literature on parallel sorting and selection~\cite{HaggkvistH81,AjtaiKS83,BollobasT83,Kruskal83, Leighton84, BollobasH85, Alon85, AlonAV86, AzarV87, Pippenger87, AlonA88a, AlonA88b, AzarP90, BollobasB90,FeigeRPU94, BravermanMW16, BravermanMP19,CohenMM20}.
Also, note that the noisy comparison models considered in some of these papers (where comparisons are incorrect with a certain probability) is different from the adversarial comparator model we study.
Thematically similar investigations on round complexity exist in the context of best arm identification for multi-armed bandits~\cite{AgarwalAAK17, TaoZZ19}.


\section{Preliminaries}
\label{sec:pre}

We recall the definition of total variation distance between probability distributions.
 \begin{definition}
   The \emph{total variation distance} between distributions $p$ and
   $q$ on $\Omega$ is defined as
   \[
     \dtv(p,q) = \max_{S \subseteq \Omega} p(S) - q(S) =\frac12 \int_{x \in \Omega} |p(x) - q(x)| d= \frac12 \|p - q\|_1 \in [0,1].
   \]
\end{definition}

  

We now define differential privacy, and the variants which we are concerned with.

\begin{definition}[\cite{DworkMNS06}]
  An algorithm $M$ with domain $\mathcal{X}^n$ is \emph{$\ve$-differentially private} if, for all $S \subseteq \mathrm{Range}(M)$ and all inputs $x,y \in \mathcal{X}^n$ which differ in exactly one point,
  \[
    \Pr[M(x) \in S] \leq e^\varepsilon \Pr[M(y) \in S].
  \]
  This definition is also called \emph{pure differential privacy}.
\end{definition}

In the local setting of differential privacy, we imagine that each user has a single datapoint. 
We require that each individual's output is differentially private.
\begin{definition}[\cite{Warner65, EvfimievskiGS03, KasiviswanathanLNRS11}]
  Suppose there are $n$ individuals, where the $i$th individual has
  datapoint $X_i$. In each round $q$ of the protocol, there is a set
  $U_q\subseteq [n]$ of active individuals, and  each individual
  $i$ in $U_q$  computes some (randomized) function of their datapoint $X_i$, and of
  all messages $\{m_{r,j}: r\le q, j \in U_r\}$ output by all individuals in previous rounds, and outputs a message $m_{q,i}$. 
  A protocol is \emph{$\varepsilon$-locally differentially private}
  (LDP) if the set $\{m_{q,i}: q \in [t], i \in U_q\}$ of all messages
  output during the $t$ rounds of the protocol is $\ve$-differentially
  private with respect to the inputs $(X_1, \ldots, X_n)$.
\end{definition}

We note that there are many notions of interactivity in LDP, and we
cover the two primary definitions which we will be concerned with:
non-interactive and sequentially interactive protocols.
\begin{definition}
  An $\ve$-LDP protocol is \emph{non-interactive} if the number of
  rounds is $t=1$, and $U_1 = [n]$, i.e., every individual $i$ outputs a
  single message $m_i$, dependent only on their datapoint $X_i$.

  An $\ve$-LDP protocol is \emph{sequentially interactive} with $t$
  rounds of interaction if the sets $U_1, \ldots, U_t$ of active individuals
  in each round are disjoint.
  
\end{definition}

We recall the canonical $\ve$-LDP algorithm, randomized response.
\begin{lemma}
  \label{lem:rr}
  \emph{Randomized response} is the protocol when each user has a bit $X_i \in \{0,1\}$ and outputs $X_i$ with probability $\frac{e^\ve}{1 + e^\ve}$ and $1 - X_i$ with probability $\frac{1}{1 + e^\ve}$.
  It satisfies $\ve$-local differential privacy.
\end{lemma}

There exists a simple folklore algorithm for $\ve$-LDP $2$-wise simple hypothesis testing: use randomized response to privately count the number of samples which fall into the region where one distribution places more mass, and output the distribution which is more consistent with the resulting estimate.
This gives the following guarantees.
\begin{lemma}
  There exists a non-interactive $\ve$-LDP algorithm which solves $2$-wise simple hypothesis testing with probability $1- \beta$, which requires $n = O(\log(1/\beta)/\alpha^2\varepsilon^2)$ samples.
\end{lemma}

This can be extended to $k$-wise simple hypothesis testing by simply running said algorithm on pairs of distributions and picking the one which never loses a hypothesis test.
This gives us an $\tilde O(k^2)$ baseline algorithm for locally private hypothesis selection.
\begin{corollary}
  There exists a non-interactive $\ve$-LDP algorithm which solves $k$-wise simple hypothesis testing with high probability, which requires $n = O(k^2\log k/\alpha^2\varepsilon^2)$ samples.
\end{corollary}
We note that the same algorithm also solves the more general problem of $\ve$-LDP hypothesis selection, see Section~\ref{sec:comparators} and particularly Section~\ref{sec:comp:easy-k2}.

\section{Lower Bounds for Locally Private Hypothesis Selection}
\label{sec:lb}

In this section we state sample complexity lower bound results on
locally private hypothesis selection. We will first focus on the lower
bound for non-interactive protocols, and leverage a known lower
bound on locally private selection due to \cite{Ullman18} (a similar statement appears in~\cite{DuchiJW17}), which also
follows from the lower bound for sparse estimation in~\cite{DuchiJW17}. Let
$d\in \mathbb{N}$, $\alpha \in [0,1]$, and let $U_d$ be a uniform
distribution over $\{\pm 1\}^d$. For every $b \in \{\pm 1\}$ and
$j\in [d]$, we define distribution
$p_{b,j} = (1 - \alpha) U_d + \alpha \left( U_d \mid x_j = b\right)$,
that is, the distribution that is uniform over $\{\pm 1\}^d$ except
that $X_j = b$ with probability $1/2 + \alpha$.

\begin{theorem}[Theorem 3.1 of \cite{Ullman18}]\label{goodhombre}
  Let $\ve \in (0,1)$.  Let $d > 32$, $B$ be distbuted uniformly
  over $\{\pm 1\}$, and let $J$ be distributed uniformly over $[d]$.
  Suppose $M$ is an non-interactive $\ve$-LDP protocol and $n$ is such that
  \[
    \Pr_{B, J, X_1, \ldots , X_n \sim (p_{B,J}| B, J)} [ M(X_1, \ldots
    , X_n) = (B, J)] \geq 1/3.
  \]
  Then
  $$n = \Omega\left( \frac{d \log d}{ \alpha^2 \ve^2}\right).$$
\end{theorem}

To obtain a lower bound on hypothesis selection, we will rely on the
following fact that bounds the total variation distance between the
distributions $p_{b,j}$ (see e.g., Lemma 6.4 in \cite{KamathLSU19}).

\begin{fact}\label{prodtv}
  Let $q$ and $q'$ be two product distributions over $\{\pm 1\}^d$
  with mean vectors $\mu$ and $\mu'$ respectively, such that
  $\mu_i \in [-1/3, 1/3]$ for all $j\in [d]$. Suppose that
  $\|\mu - \mu'\|_2 \geq \alpha$ for any $\alpha \leq \alpha_0$ with
  some absolute constant $0 < \alpha_0 \leq 1$. Then
  $\dtv(q, q') \geq C \alpha$, for some absolute constant $C$.
\end{fact}

\begin{theorem}[Non-interactive lower bound]
  Let $\ve \in (0,1)$.
  Suppose $M$ is a non-interactive an $\ve$-LDP protocol that solves the $k$-wise
  simple hypothesis testing problem  with probability at least $1/3$
  when given $n$ samples from some distribution $p \in \mathcal{Q}$,
  where $\mathcal{Q} = \{q_1, \ldots, q_k\}$ are distributions such
  that $\min_{i \neq j} \dtv(q_i, q_j) \ge \alpha$. Then
  $$n \geq \Omega\left( \frac{k \log k}{ \alpha^2 \ve^2}\right).$$
\end{theorem}

\begin{proof}
  Let $\mathcal{Q} = \{p_{b,j} \mid b\in \{\pm 1\}, j\in [d]\}$ be a
  set of $k=2d$ probability distributions. For any pair of
  distributions $q, q'\in \mathcal{Q}$, we know from Fact~\ref{prodtv}
  that $\dtv(q, q') \geq \alpha/C$ for some absolute constant $C$.
  Thus, our stated bound follows from
  Theorem~\ref{goodhombre}.
\end{proof}

Next we will derive a sample complexity lower bound for general
locally private protocols. We will build on a result due to
\cite{DuchiR19} and consider the set of 1-sparse Gaussian
distributions $\{\mathcal{N}(\theta, I_d) \mid \theta\in \Theta\}$,
where
$\Theta = \{\theta \in \mathbb{R}^d\mid \|\theta\|_2 = \alpha,
\|\theta\|_0 = 1\}$ is the set of vectors that have a single non-zero
coordinate, equal to $-\alpha$ or $+\alpha$.

Following the result of \cite{DuchiR19} (and the framework of
\cite{BravermenGMNW}), we can obtain a general lower bound analogous
to Theorem~\ref{goodhombre}.

\begin{theorem}[Corollary~6 of \cite{DuchiR19}, Theorem~4.5 of \cite{BravermenGMNW}] \label{goodhombre2}
  Let $\ve \in (0,1)$. Let $U$ be a uniform distbution over
  $\Theta$. Suppose $M$ is an $\ve$-LDP protocol, and $n$ is such that
  \[
    \Pr_{\theta \sim U, X_1, \ldots , X_n \sim \mathcal{N}(\theta, I)}
    [ M(X_1, \ldots , X_n) = \theta] \geq 1/3.
  \]
  Then
  $$n \geq \Omega\left( \frac{d}{ \alpha^2 \ve^2}\right).$$
\end{theorem}


\lb*
\begin{proof}
  For any two 1-sparse vectors $\theta, \theta'\in \Theta$ such that
  $\theta\neq \theta'$, the total variation distance between their
  Gaussian distributions is given by
  $\|\theta - \theta'\|_2 = \sqrt{2} \alpha$ (see, e.g.,
  \cite{DevroyeMR18b}). Thus, our stated bound follows from
  Theorem~\ref{goodhombre2}.
\end{proof}


\section{Non-Interactive Locally Private Hypothesis Selection}
\label{sec:non-adaptive}

In this section, we prove Theorem \ref{thm:1Round}.  For simplicity of notation, we assume without loss of generality that $q_1,q_2,\dots,q_k$ and $p$ are discrete probability distributions 
on domain $[N]$, where $[N]:= \{1,2,\ldots, N\}.$ 
See the discussion in Remark~\ref{rmk:cont} on how to deal with continuous distributions.
Here we propose an algorithm which uses $n\lesssim k\,\polylog(k)/(\alpha^4\eps^2)$ samples, and outputs a distribution $\hat q \in \mathcal{Q}$ which has TV distance of at most $O(\alpha)$ with $p$, when $\beta \ll  \alpha^2/\log k.$ Recall that $\beta :=  \min_{q \in \mathcal{Q}} \dtv(p, q)$.
In this mechanism, the users are divided into $k$ groups $G_1,G_2,\dots,G_k$ of size $n/k$ each. Let $X_{ij} \sim p$ denote the sample with the $j^{th}$ user in the group $G_i$. Our non-interactive mechanism is described in Algorithm \ref{alg:LikelyhoodTest}.

\begin{algorithm}
 \caption{Non-interactive $\eps$-DP mechanism for LPHS}
 \label{alg:LikelyhoodTest}
\hspace*{\algorithmicindent} \textbf{Input:} Distributions $\mathcal{Q} = \{q_1,\dots,q_k\}$, Samples $(X_{ij})_{i\in [k],j \in [n/k]}$ from unknown distribution $p$, sensitivity parameter for Laplace noise $L$, privacy parameter $\ve$, function $\gamma:[N]\to \R^+$ such that $|\log(\gamma(a)/q_i(a))|\le L$ for all $a\in [N],i\in [k].$\footnotemark \\
\hspace*{\algorithmicindent} \textbf{Output:} $\hat q \in \mathcal{Q}$ such that $\dtv(p,\hat q)\le \alpha$ with high probability.
\begin{algorithmic}[1]
 \For {$i\in [k]$}
  \For {$j \in [n/k]$}
    \State The $j^{th}$ user in group $G_i$ sends $Z_{ij} := \log(\gamma(X_{ij})/q_i(X_{ij}))+ \Lap(L/\eps)$ to the central server
  \EndFor
  \State The central server computes $C_i= \frac{1}{(n/k)} \cdot \sum_{j\in [n/k]} Z_{ij}.$
 \EndFor
 \State \textbf{return} $\argmin_i C_i$.
\end{algorithmic}
\end{algorithm}
\footnotetext{In other words, we require $\Dinfty(\gamma||q_i),\Dinfty(q_i||\gamma)\le L$ for all $i\in [k]$, i.e., all the distributions $q_1,q_2,\dots,q_k$ are close to some distribution $\gamma$. To prove Theorem~\ref{thm:1Round}, we will instantiate Algorithm~\ref{alg:LikelyhoodTest} with $\gamma$ being the uniform distribution on $[N]$, but we state Algorithm~\ref{alg:LikelyhoodTest} with arbitrary $\gamma$ for generality.}

\begin{lemma}
\label{thm:LikelyhoodTest}
Let $\eps\in (0,1]$ be some fixed privacy parameter. Suppose
$\beta\ll \alpha^2/L$ and $ n\gg \frac{ k(\log k) L^2} {\alpha^4 \eps^2}$. Then Algorithm~\ref{alg:LikelyhoodTest} is $\eps$-LDP and outputs
$\hat q \in \mathcal{Q}$ with probability at least $1-1/k^2$ such that
$\dtv(p, \hat q) \leq \alpha$.
\end{lemma}
\begin{proof}
  From our assumption, $|\log(\gamma(a)/q_i(a))|\le L$ for
  $a\in [N],i\in [k]$. The algorithm adds noise sampled from
  $\Lap(L/\eps)$, hence $\eps$-LDP guarantee follows easily from the
  properties of the Laplace mechanism \cite{DworkR14}. We will now
  prove correctness.  Let $i^*\in [k]$ be such that
  $\dtv(p, q_{i^*}) = \beta$.  Fix a group $G_i$ and consider,

  \begin{eqnarray*}
  \E_{X_{ij}\sim p}[C_i] &=& \frac{1}{(n/k)} \cdot \E \left[\sum_{j\in [n/k]} Z_{ij} \right]\\
  &=& \E_{a\sim p}\left[\log\left(\frac{\gamma(a)}{q_i(a)}\right)\right] \quad ({\text{By the linearity of expectation and}} \hspace{1mm} \E[\text{Lap}(L/\epsilon)] = 0)\\
  &=&\sum_{a\in [N]} p(a)\log\left(\frac{\gamma(a)}{q_i(a)}\right) \\
  &=& \sum_{a \in [N]} q_{i^*}(a)\log\left(\frac{q_{i^*}(a)}{q_i(a)}\right) + \sum_{a \in [N]} (p(a)-q_{i^*}(a))\log\left(\frac{\gamma(a)}{q_i(a)}\right) +\sum_{a \in [N]} q_{i^*}(a)\log\left(\frac{\gamma(a)}{q_{i^*}(a)}\right) \\
  &=& \DKL(q_{i^*}||q_i) + \sum_{a \in [N]} (p(a)-q_{i^*}(a))\log\left(\frac{\gamma(a)}{q_i(a)}\right) - \DKL(q_{i^*}||\gamma).
  \end{eqnarray*}

Let $B=-\DKL(q_{i^*}||\gamma)$. By re-arranging the above term we get 
\begin{eqnarray*}
\left|\E_{X_{ij}\sim q}[C_i]-\DKL(q_{i^*}||q_i)-B \right| &\le&  \sum_{a \in [N]} |p(a)-q_{i^*}(a)| \cdot \left|\log\left({\gamma(a)}/{q_i(a)}\right)\right| \\
&\le&  \sup_{a \in [N]}\left|\log\left(\frac{\gamma(a)}{q_i(a)}\right)\right| \cdot \left( \sum_{a \in [N]} |p(a)-q_{i^*}(a)|\right) \\
&\le& L\cdot 2 \dtv(p,q_{i^*}) \\
&\le& 2L\beta \le 0.1\alpha^2.
\end{eqnarray*}
Now observe that each $Z_{ij}$ can be expressed as $W_{ij} + Y_{ij}$, where $Y_{ij} \sim \text{Lap}(L/\epsilon)$, and the support of random variable $W_{ij}$ is in the interval $[-L, L]$ from our assumption.
Therefore, we can apply the standard Hoeffding's inequality and concentration of Laplace random variables (see \cite{hoeffding1994probability, chan2011private} for example) to obtain
$\Pr\left[|C_i-\E[C_i]|\ge 0.1 \alpha^2\right]\le \exp\left(-\Omega(1)\cdot \frac{(n/k)\alpha^4}{(L/\eps)^2}\right) \le \frac{1}{k^3}$.

By taking the union bound, with probability at least $1-1/k^2$, $|C_i-\DKL(q_{i^*}||q_i)-B|\le 0.2\alpha^2$  for all $i\in [k]$.
In particular, $C_{i^*}\le B+0.2\alpha^2$. This implies that if $i'=\argmin_i C_i$, then $C_{i'}\le B+0.2\alpha^2$. It remains to argue that $\dtv(p, q_{i'}) < \alpha$. Suppose not. Consider any $q_i$ such that $\dtv(p, q_i) > \alpha$. This implies that $\dtv(q_{i*}, q_i) > \alpha/2$ based on our assumption. Now consider
$C_i \ge B+\DKL(q_{i^*}||q_i) - 0.2 \alpha^2 \ge B+ 2\dtv(q_{i^*}, q_i)^2 - 0.2\alpha^2 \ge B+0.3\alpha^2$, where we used Pinsker's inequality.
\end{proof}

We will now prove that we can take $L=O(\log k)$ in Algorithm~\ref{alg:LikelyhoodTest} and Lemma~\ref{thm:LikelyhoodTest}. For this we will need the following lemma. Given a randomized map\footnote{i.e., $\phi(a)$ has a distribution over $[N']$ for each $a\in [N]$.} $\phi:[N]\to [N']$ and a distribution $q$ on $[N]$, the distribution $\phi \circ q$ on $[N']$ is defined as the distribution of $\phi(a)$ when $a$ is sampled from $q$. (In other words, $\phi \circ q$ is the pushforward of $q$.) For the remaining part of this section, let $U_{N'}$ denote the uniform distribution on $[N'].$

\begin{lemma}[Flattening Lemma]
\label{lem:splitting}
Let $q_1,q_2,\dots,q_k$ be distributions over $[N]$. There exists a randomized map $\phi:[N]\to [N']$ (depending on $q_1,\dots,q_k$) for some $N\le N'\le (k+1)N$ s.t.
\begin{enumerate}
    \item for every $a\in [N'], i\in [k],$ $\frac{1}{2N'}\le (\phi \circ q_i)(a) \le \frac{1}{N}$ and
    \item $\dtv(\phi\circ q_i, \phi \circ q_{i'}) =\frac{1}{2} \cdot \dtv(q_i, q_{i'})$ for any two distributions $q_i, q_{i'}$.
\end{enumerate}
\end{lemma}
\begin{proof}
Let $M(a)=\max_{i\in [k]} q_i(a)$ for $a\in [N]$. Let $N'=\sum_{a\in [N]} \lceil M(a)\cdot N\rceil$ and let $[N']=\cup_{a\in [N']} S_a$ be a partition of $[N']$ with $|S_a|=\lceil M(a) \cdot N\rceil$. Define $\phi':[N]\to [N']$ as follows: $\phi'(a)$ is uniformly distributed over $S_a$. Now it is clear that for every for every $b\in [N']$, $(\phi'\circ q_i)(b)\le \frac{1}{N}$. It is also clear that $\ellone{\phi'\circ q_{i}- \phi'\circ q_{i'}}=\ellone{q_{i}-q_{i'}}$ for any two distributions $q_{i},q_{i'}$.
We now mix in the uniform distribution $U_{N'}$ into $\phi'$, i.e., we define $\phi:[N]\to [N']$ as follows: $\phi(a)$ is distributed as $\phi'(a)$ with probability $1/2$ and distributed as $U_{N'}$ with probability $1/2$. Now for every $b\in [N'],$ $\frac{1}{2N'}\le (\phi\circ p_i)(b) \le \frac{1}{N}$. And $\ellone{\phi'\circ q_i- \phi'\circ q_{i'}}=\frac{1}{2}\ellone{q_i-q_{i'}}$ for any two distributions $q_i,q_{i'}$. 
We are now left with showing the upper bound on $N.$
\begin{eqnarray*}
N' &=& \sum_{a\in [N]}\lceil M(a) \cdot N\rceil \\
&\le& \sum_{a\in [N]} (M(a) \cdot N +1) \\
&=& N +\sum_{a\in [N]} \left(\max_{i\in [k]} q_i(a)\right) N \\
&\le& N+\sum_{a\in [N]} \left(\sum_{i\in [k]} q_i(a)\right)N =(k+1)N.
\end{eqnarray*}
\end{proof}
Now we have all the ingredients to finish the proof of Theorem \ref{thm:1Round}.
\begin{proof}
By using the randomized map $\phi$ as constructed in Lemma~\ref{lem:splitting}, the users first map their sample $a\sim p$ to a sample $\phi(a)$. Note that $\phi(a) \sim \phi\circ p$. Next we run the Algorithm~\ref{alg:LikelyhoodTest} with distributions $\phi\circ q_1,\dots, \phi\circ q_k$ and $\gamma:[N']\to \R^+$ given by $\gamma(b)=1/N'$ for all $b \in [N']$. From the first property mentioned in Lemma~\ref{lem:splitting}, we get $L=\log(k)+O(1).$ From the second property in Lemma~\ref{lem:splitting}, we know the TV distances are preserved by $\phi$.
\end{proof}

\begin{remark} 
If we are able to get $L=O(\alpha)$, then we get the nearly optimal sample complexity of $n=O\left(\frac{k\polylog(k)}{\alpha^2\eps^2}\right)$, formalized in the following question.
\end{remark}

\begin{question}
Given distributions $p_1,p_2,\dots,p_k$ which are $\alpha$-far to each other in $\ell_1$-distance, is there a randomized map $\phi:[n]\to [N]$ (which can depend on $p_1,\dots,p_k$) s.t.
\begin{enumerate}
    \item For all $i\in [n], a\in [N]$,  $\frac{1-\alpha}{N}\le (\phi \circ p_i)(a) \le \frac{1+\alpha}{N}$ and
    \item $\ellone{\phi\circ p- \phi \circ q} = \Theta(\ellone{p-q})$ for any two distributions $p,q$ with high probability.
\end{enumerate}
Note that $N$ can be arbitrarily large.
\end{question}

\begin{remark}
\label{rmk:cont}
The arguments in our proof can be easily generalized to continuous probability distributions. However, as our results do not depend on the domain size, it is intuitive to think of the following simple mapping from continuous distributions to discrete distributions on the domain $[N]$. First, we can approximate (to any precision) a set of continuous distributions by a set of discrete distributions on a finite support such that TV distances are preserved. We can then map any set of discrete distributions on a finite support to a set of discrete distributions on the domain $[N]$, where $N$ will depend on the desired precision. 
\end{remark}
 
\section{Hypothesis Selection via Adversarial Comparators}
\label{sec:comparators}
In this section, we give upper bounds for locally private hypothesis selection via a reduction to adversarial comparators, as introduced by~\cite{AcharyaJOS14b,AcharyaFJOS18}.
We begin by describing the reduction and how it can be implemented in the LDP setting in Section~\ref{sec:comp:reduction}.
This allows us to immediately obtain a non-interactive private algorithm which takes $\tilde O(k^2)$ samples and a sequentially-interactive algorithm which takes $\tilde O(k)$ samples (Section~\ref{sec:comp:easy-k2}).
However, this sequentially-interactive algorithm requires $O(\log k)$ rounds -- we give an algorithm which improves upon this round-complexity by an exponential factor.
We start in Section~\ref{sec:comp:2-rounds} by giving a simple $\tilde O(k^{4/3})$-sample algorithm which takes $2$ rounds: with the addition of only a single additional round, the sample complexity becomes significantly subquadratic.
This illustrates one of the main ideas behind our full upper bound, an $\tilde O(k)$-sample algorithm which takes only $O(\log \log k)$ rounds.
This is acheived by generalizing our $2$ round algorithm to general $t$: we give $t$-round algorithms for $1 \leq t \leq O(\log \log k)$, with sample complexities which interpolate between $\tilde O(k^2)$ and $\tilde O(k)$.
Other ideas are required to achieve an approximation which does not increase with $t$, which are described in Section~\ref{sec:comp:ub}.
We complement these upper bounds with lower bounds which show that these algorithms in the adversarial comparator setting are essentially tight (for \emph{every} choice of $t$) (Section~\ref{sec:comp:lb}).

\subsection{Adversarial Comparators and Connections to Locally Private Hypothesis Selection}
\label{sec:comp:reduction}
We describe the adversarial comparator setting of~\cite{AcharyaJOS14b,AcharyaFJOS18}, as well as their reduction to this model for the hypothesis selection problem.
The input is a set of $k$ items, with unknown values $x_1, \dots, x_k \in \mathbb{R}$.
An adversarial comparator is a function $C$, which takes two items $x_i$ and $x_j$,\footnote{In a slight abuse of notation, we use $x_i$ to refer to the item as well as its value.}, and outputs $\max\{x_i, x_j\}$ if $|x_i - x_j| > 1$ and $x_i$ or $x_j$ (adversarially) if $|x_i - x_j| \leq 1$.

We note that such a comparator can be either non-adaptive or adaptive.
In the former case, the results of all comparisons must be fixed ahead of time, whereas in the latter case, results of comparisons may depend on previous comparisons.
All of the mentioned algorithms will work in the (harder) adaptive case, and our lower bounds are for the (easier) non-adaptive case, and thus both have the same implications in the alternate setting for adaptivity.

We sometimes denote a comparison as a \emph{query}.
The goal is to output an item with value as close to the maximum as possible, with probability at least $2/3$.\footnote{Usual arguments allow us to boost this success probability to $1-\beta$ at a cost of $O(\log (1/\beta))$ repetitions, which can be done in parallel.}
More precisely, let $x^* = \max \{x_1, \dots, x_k\}$.
A number $x$ is a $\tau$-approximation of $x^*$ if $x \geq x^* - \tau$.
Simple examples (e.g., Lemma 2 of~\cite{AcharyaFJOS18}) show that it is impossible to output a $\tau$-approximation with probability $\geq 2/3$ for any $\tau < 2$ when we have $k \geq 3$ items.

We initiate study of \emph{parallel} approximate maximum selection under adversarial comparators.
Parallel maximum selection has recently been studied in other settings (including the standard comparison setting and with noisy (but not adversarial) comparisons, see, e.g.,~\cite{BravermanMW16}).
In this setting, the algorithm has $t$ rounds: in round $i$, the algorithm simultaneously submits $m_i$ pairs of items, and then simultaneously receives the results of the adversarial comparator applied to all $m_i$ pairs.
The total query complexity is $\sum_{i = 1}^t m_i$.

We now discuss the connection between this problem and hypothesis selection, as presented in Section 6 of~\cite{AcharyaFJOS18}.
We will then show how this connection still applies when considering the same problem under LDP.
First, we recall the Scheff\'e test of Devroye and Lugosi~\cite{DevroyeL01}, as described in Algorithm~\ref{alg:scheffe}.\footnote{We comment that this can be implemented in near-linear time, and $q_1(S)$ and $q_2(S)$ can be estimated to sufficient accuracy using Monte Carlo techniques.}
Given $n$ samples from $p$, with probability at least $1 - \beta$, it will output a distribution $\hat q$ such that $\dtv(p,\hat q) \leq 3 \min\{\dtv(p,q_1), \dtv(p,q_2)\} + \sqrt{\frac{2.5\log (1/\beta)}{n}}$.
In other words, if $ \min \Paren{ \dtv(p, q_1),  \dtv(p, q_2)} \leq \alpha$, then $n = O\left(\frac{\log(1/\beta)}{\alpha^2}\right)$ samples suffice to output a $\hat q \in \{q_1, q_2\}$ such that $\dtv(p,\hat q) \leq (3 + \gamma)\alpha$, where $\gamma$ can be taken to be an arbitrarily small constant. Another way to phrase this is that the test returns $q_1$ if $\dtv(p,q_1) < \frac{1}{3 + \gamma}\dtv(p,q_2)$, it returns $q_2$ if $\dtv(p,q_2) < \frac{1}{3 + \gamma}\dtv(p,q_1)$, and it may return arbitrarily otherwise.
If we let $x_i = -\log_{3+\gamma}\dtv(p,q_i)$, then the test will output $\max\{x_i, x_j\}$ if $|x_i - x_j| > 1$, or arbitrarily otherwise.
Note that this is precisely an implementation of the adversarial comparator function $C$ as described above, and thus the hypothesis selection problem can be reduced to (approximate) maximum selection with adversarial comparators.
In particular, a $\tau$-approximation for the maximum selection problem becomes a $(3+\gamma)^\tau$ agnostic approximation factor for hypothesis selection, which becomes $3^\tau + \gamma'$ if $\tau$ is a constant, for some other constant $\gamma' >0$ which can be taken to be arbitrarily small.
Each comparison is implemented using $O\left(\frac{\log (1/\beta)}{\alpha^2}\right)$ samples from $p$-- in fact, by a union bound argument, if we wish to perform $m$ comparisons and require the total failure probability under $1/3$, all of them can be done with the same set of $O\left(\frac{\log m}{\alpha^2}\right)$ samples.

\begin{algorithm}
    \caption{Scheff\'e Test}
    \label{alg:scheffe}
    \hspace*{\algorithmicindent} \textbf{Input:} $n$ samples $X_1, \dots, X_n$ from \emph{unknown} $p$, distributions $q_1$ and $q_2$ \\
    \hspace*{\algorithmicindent} \textbf{Output:} Distribution $q_1$ or $q_2$
    \begin{algorithmic}[1] 
        \Procedure{Scheff\'e}{$X, q_1, q_2$} 
        \State Let $S = \{x\ :\ q_1(x) > q_2(x)\}$.
        \State Let $q_1(S)$ and $q_2(S)$ be the probability mass that $q_1$ and $q_2$ assign to $S$ .
        \State Let $\hat p(S) = \frac{1}{n}\sum_{i=1}^n \mathbbm{1}_{X_i \in S} $ be the empirical mass assigned by $X_1, \dots, X_n$ to $S$. \label{ln:emp-mass}
        \If {$|q_1(S) - \hat p(S)| < |q_2(S) - \hat p(S)|$}
        \State \textbf{return} $q_1$.
        \Else
        \State \textbf{return} $q_2$.
        \EndIf
        \EndProcedure
    \end{algorithmic}
\end{algorithm}

It remains to justify that a similar reduction still holds under LDP constraints.
Recall that each individual possesses a single $X_i$, and they wish for their messages sent to the curator to be $\ve$-DP.
Only Line~\ref{ln:emp-mass} of Algorithm~\ref{alg:scheffe} depends on the private data, which is a statistical query, easily implemented under LDP.
More precisely, rather than sending the bit $\mathbbm{1}_{X_i \in S}$ to the curator, the user can send $Y_i$, which is a version of it privatized by Randomized response (Lemma~\ref{lem:rr}).
The curator can then form an $\ve$-LDP estimate of $p(S)$ by computing $\hat p(S) = \frac{e^\ve + 1}{e^\ve - 1}\left(\frac{1}{n}\sum Y_i - \frac1{e^\ve+1}\right)$.
Plugging this estimate into Line~\ref{ln:emp-mass}, it is not hard to show the modified procedure satisfies the following accuracy guarantee:
if $\min \Paren{ \dtv(p, q_1),  \dtv(p, q_2)} \leq \alpha$, then $n = O\left(\frac{\log(1/\beta)}{\ve^2\alpha^2}\right)$ samples suffice to output an $\ve$-LDP $\hat q \in \{q_1, q_2\}$ such that $\dtv(p,\hat q) \leq (3 + \gamma)\alpha$, where $\gamma$ can be taken to be an arbitrarily small constant.

The above addresses the case of a single comparison.
If we wish to make $m$ comparisons (which are all correct with high probability), we partition users into $m$ sets of size $O\left(\frac{\log m}{\ve^2\alpha^2}\right)$ and use the data from each part to privately perform the appropriate comparison.
This takes a total of $O\left(\frac{m \log m}{\ve^2\alpha^2}\right)$ samples.
In particular, we can not reuse the same set of $O(\log m)$ samples for all comparisons (as in the non-private case), since it violate the privacy constraint, and doing so would give rise to algorithms which violate our main lower bound for locally private hypothesis selection (Theorem~\ref{thm:fully-interactive-lb}).
Finally, we note that a $t$-round algorithm in the maximum selection setting corresponds to a $t$-round sequentially interactive $\ve$-LDP algorithm for hypothesis selection, as we never query the same individual twice.

To conclude this section, we state the guarantees of the (trivial) algorithm which performs maximum selection from a set of $2$ elements, and the corollary for LDP hypothesis selection implied by the above reduction.

\begin{claim}
  There exists a $1$-round algorithm which achieves a $1$-approximation in the problem of parallel approximate maximum selection with adversarial comparators, in the special case where $k= 2$.
  The algorithm requires $1$ query.
\end{claim}

\begin{corollary}
  There exists a $1$-round algorithm which achieves a $(3+\gamma)$-agnostic approximation factor for locally private hypothesis selection with probability $1 -\beta$, in the special case where $k=2$, where $\gamma > 0$ is an arbitrarily small constant.
  The sample complexity of the algorithm is $O\left(\frac{\log (1/\beta)}{\ve^2 \alpha^2}\right)$.
\end{corollary}

For the following subsections, we will focus on the problem of parallel approximate maximum selection with adversarial comparators, stating corollaries to locally private hypothesis selection as appropriate.
Our primary concerns will be to simultaneously minimize the query/sample complexity and the round complexity, while minimizing the approximation/agnostic approximation factor is a secondary concern.
Nevertheless, our new algorithms for maximum selection will have an approximation constant of at most $3$, very close to the information-theoretic optimum of $2$.

\subsection{Baseline Algorithms}
\label{sec:comp:easy-k2}
In this section, we state some baseline results in this model, based on previously known algorithms.
This includes a $O(k^2)$-query non-interactive algorithm, and a $O(k)$-query $O(\log k)$-round algorithm.

The first method is a ``round-robin'' tournament method, which, in a single round, performs all pairwise comparisons and outputs the item which is declared to be the maximum the largest number of times (Algorithm~\ref{alg:round-robin}). 
This straightforward method is stated and analyzed in~\cite{AcharyaJOS14b,AcharyaFJOS18}, and the equivalent procedure for hypothesis selection (absent privacy constraints) was known prior~\cite{DevroyeL01}.
\begin{algorithm}
    \caption{1-Round Algorithm for Maximum Selection}
    \label{alg:round-robin}
    \hspace*{\algorithmicindent} \textbf{Input:} $k$ items $x_1, \dots, x_k$ \\
    \hspace*{\algorithmicindent} \textbf{Output:} Approximate maximum $x_i$ 
    \begin{algorithmic}[1] 
        \Procedure{Round-Robin}{$x_1, \dots, x_k$} 
      \For {all pairs $x_i, x_j$}
        \State Compare $x_i$ and $x_j$, record which one is reported to be the winner.
      \EndFor
      \State \textbf{return} the $x_i$ which is reported to be the winner the most times.
        \EndProcedure
    \end{algorithmic}
\end{algorithm}

\begin{claim}
  \label{clm:round-robin}
  There exists a $1$-round algorithm which achieves a $2$-approximation in the problem of parallel approximate maximum selection with adversarial comparators.
  The algorithm requires $O(k^2)$ queries.
\end{claim}

\begin{corollary}
  There exists a $1$-round algorithm which achieves a $(9+\gamma)$-agnostic approximation factor for locally private hypothesis selection with high probability, where $\gamma > 0$ is an arbitrarily small constant.
  The sample complexity of the algorithm is $O\left(\frac{k^2 \log k}{\ve^2 \alpha^2}\right)$.
\end{corollary}

The clear drawback of this method is that the complexity of the resulting algorithms is quadratic in $k$.
Unfortunately, a simple argument shows that this is tight for any $1$-round protocol: roughly, if we do not compare the smallest and second smallest items, we do not know which is smaller, and thus any algorithm which doesn't perform all $\binom{k}{2}$ comparisons in its $1$ round will be wrong with probability $1/2$ (more formal lower bounds for more general settings appear in Section~\ref{sec:comp:lb}).
The natural questions are, if we expend more rounds, can we reduce the sample complexity? 
And how many rounds are needed to achieve the information-theoretic optimum of a linear query complexity?
Many recent works have focused on this question without concern for the number of rounds expended~\cite{AjtaiFHN09, DaskalakisK14, SureshOAJ14, AcharyaJOS14b, AcharyaFJOS18}, culminating in algorithms with linear complexity.
When the round complexity is analyzed, it can be shown that all these methods take $O(\log k)$ rounds.
We state the implied results for our setting in the following claim and corollary, omitting details as we will shortly improve on the round complexity to be $O(\log \log k)$.

\begin{claim}[\cite{AcharyaJOS14b, AcharyaFJOS18}]
  There exists an $O(\log k)$-round algorithm which achieves a $2$-approximation in the problem of parallel approximate maximum selection with adversarial comparators.
  The algorithm requires $O(k)$ queries.
\end{claim}

\begin{corollary}
  There exists an $O(\log k)$-round algorithm which achieves a $(9+\gamma)$-agnostic approximation factor for locally private hypothesis selection with high probability, where $\gamma > 0$ is an arbitrarily small constant.
  The sample complexity of the algorithm is $O\left(\frac{k \log k}{\ve^2\alpha^2}\right)$.
\end{corollary}

\subsection{A Sub-Quadratic Algorithm with $2$ Rounds}
\label{sec:comp:2-rounds}
In this section, we give a simple $2$-round algorithm which results in a significantly better query complexity of $O(k^{4/3})$.
In Section~\ref{sec:comp:ub}, we generalize this to $t$-round protocols, but provide this as a warm-up and to convey one of the main ideas.

\begin{algorithm}
    \caption{2-Round Algorithm for Maximum Selection}
    \label{alg:2-round}
    \hspace*{\algorithmicindent} \textbf{Input:} $k$ items $x_1, \dots, x_k$ \\
    \hspace*{\algorithmicindent} \textbf{Output:} Approximate maximum $x_i$ 
    \begin{algorithmic}[1] 
        \Procedure{2-Round}{$x_1, \dots, x_k$} 
        \State Partition $x_1$ through $x_k$ into $k^{2/3}$ sets of size $k^{1/3}$. \label{ln:2-round-1}
        \State Run \textsc{Round-Robin} on each set to obtain $k^{2/3}$ winners. \label{ln:2-round-2}
        \State \textbf{return} the winner of \textsc{Round-Robin} on the set of $k^{2/3}$ winners. \label{ln:2-round-3}
        \EndProcedure
    \end{algorithmic}
\end{algorithm}

\begin{figure}[h!]
  \includegraphics[scale=0.5]{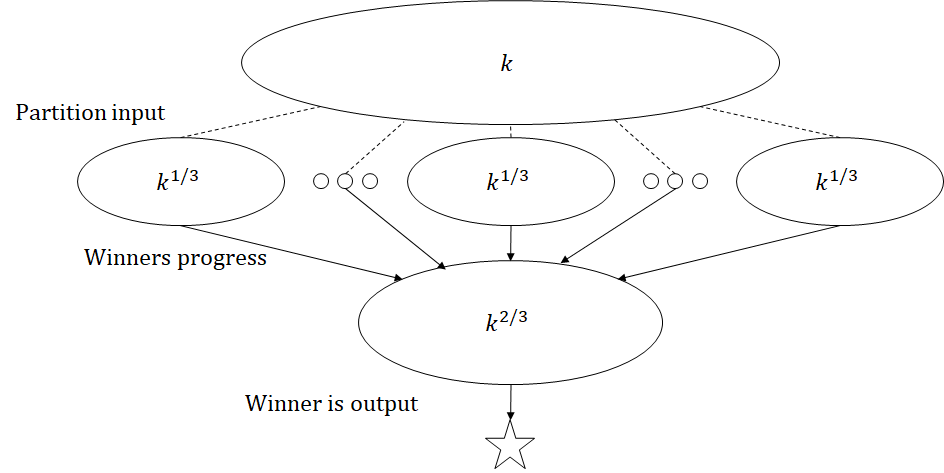}
  \centering
  \caption{An illustration of Algorithm~\ref{alg:2-round}. In the first round, the input is partitioned into sets of size $k^{1/3}$ and a round-robin tournament is performed on each.
  In the second round, a single round-robin tournament is performed on the winners from the previous round.}
\end{figure}

Algorithm~\ref{alg:2-round} describes the procedure, whose guarantees are summarized in the following theorem.
\begin{theorem}
  \label{thm:2-round}
  There exists a $2$-round algorithm which achieves a $4$-approximation in the problem of parallel approximate maximum selection with adversarial comparators.
  The algorithm requires $O(k^{4/3})$ queries.
\end{theorem}

The resulting corollary for LDP hypothesis selection is the following.
\begin{corollary}
  There exists an $2$-round algorithm which achieves a $(81+\gamma)$-agnostic approximation factor for locally private hypothesis selection with high probability, where $\gamma > 0$ is an arbitrarily small constant.
  The sample complexity of the algorithm is $O\left(\frac{k^{4/3} \log k}{\ve^2\alpha^2}\right)$.
\end{corollary}

We proceed to prove the guarantees stated in Theorem~\ref{thm:2-round}.

\begin{proof}
  The number of rounds is easily seen to be 2: Lines~\ref{ln:2-round-1} and~\ref{ln:2-round-2} can be performed in one round, and Line~\ref{ln:2-round-3}, which depends on the results of the previous round, is performed in the second round.

  We next analyze the number of queries. 
  Line~\ref{ln:2-round-2} performs the quadratic round-robin tournament of Claim~\ref{clm:round-robin} on sets of size $k^{1/3}$.
  The resulting number of queries for each set is $O(k^{2/3})$, and since there are $k^{2/3}$ sets, the total number of queries here is $O(k^{4/3})$.
  Line~\ref{ln:2-round-3} performs the same quadratic round-robin tournament on one set of size $k^{2/3}$, which takes $O(k^{4/3})$ queries.
  Therefore, the total number of queries is $O(k^{4/3})$.

  Finally, we justify that this achieves a $4$-approximation to the maximum.
  Consider the first round: a maximum element is placed into one of the $k^{2/3}$ sets, and by the guarantees of Claim~\ref{clm:round-robin}, the winner for this set will be a $2$-approximation to the maximum.
  Therefore, the maximum among the winners is a $2$-approximation to the overall maximum, and again by the guarantees of Claim~\ref{clm:round-robin}, the winner of this round will be a $4$-approximation to the maximum, as desired.
\end{proof}

\subsection{A Near-Linear-Sample Algorithm with $O(\log \log k)$ Rounds}
\label{sec:comp:ub}
In this section, we describe our main result in this setting, a family of algorithms for approximate maximum selection parameterized by $t$, which is the allowed number of rounds.
By setting $t = O(\log \log k)$, we will get an $O(k \log \log k)$-query algorithm which requires only $O(\log \log k)$ rounds, improving exponentially on the round complexity of previous approaches.
In particular, the following corollaries are obtained from Theorem~\ref{thm:better-t-round} and Corollary~\ref{cor:better-t-round-ldp} with an optimized setting of parameters.

\begin{corollary}
  \label{cor:set-params-select}
  There exists an $O(\log \log k)$-round algorithm which, with probability $9/10$, achieves a $3$-approximation in the problem of parallel approximate maximum selection with adversarial comparators.
  The algorithm requires $O(k \log \log k)$ queries.
\end{corollary}

\begin{corollary}
  \label{cor:set-params-ldp}
  There exists an $O(\log \log k)$-round algorithm which achieves a $(27+\gamma)$-agnostic factor for locally private hypothesis selection with probability $9/10$, where $\gamma > 0$ is an arbitrarily small constant.
  The sample complexity of the algorithm is $O\left(\frac{k\log k \log \log k}{\ve^2\alpha^2}\right)$.
\end{corollary}

The method is a careful recursive application of the approach described in Algorithm~\ref{alg:2-round}. 
Specifically, given $t$ allowed rounds of adaptivity, we partition the items into several smaller sets, perform the round-robin algorithm on each, and then feed the winners into the algorithm which is allowed $t-1$ rounds of adaptivity.
A judicious setting of parameters will allow the number of comparisons to decay quite rapidly as the number of rounds is increased.
This construction is described and analyzed in Section~\ref{sec:comp:ub:recursive}.
One challenge is that each round of the algorithm will potentially lose an additive $2$ in the approximation, resulting in an overall $2t$-approximation. 
To avoid this, we employ ideas from~\cite{DaskalakisK14}: we simultaneously apply two algorithms, at least one of which will be effective depending on whether the density of elements close to the maximum is high or low.
We describe the necessary modification and analyze the resulting approach in Section~\ref{sec:comp:ub:bound}.

\subsubsection{A Recursive Application of the $2$-Round Method}
\label{sec:comp:ub:recursive}
Our main result of this section will be the following lemma.
While the round and query complexity are essentially optimal (see Section~\ref{sec:comp:lb}), the quality of approximation is unsatisfactory -- our approach to improving this approximation is described in~\ref{sec:comp:ub:bound}.
\begin{lemma}
  \label{lem:t-round}
  There exists a $t$-round algorithm which achieves a $2t$-approximation in the problem of parallel approximate maximum selection with adversarial comparators.
  The algorithm requires $O(k^{1 + \frac{1}{2^t-1}}t)$ queries.
\end{lemma}

The method is described in Algorithm~\ref{alg:t-round}.
Note that for $t=1$ or $t=2$, this simplifies to Algorithms~\ref{alg:round-robin} and~\ref{alg:2-round}, respectively.
\begin{algorithm}
    \caption{$t$-Round Algorithm for Maximum Selection}
    \label{alg:t-round}
    \hspace*{\algorithmicindent} \textbf{Input:} $k$ items $x_1, \dots, x_k$, number of rounds $t$\\
    \hspace*{\algorithmicindent} \textbf{Output:} Approximate maximum $x_i$ 
    \begin{algorithmic}[1] 
        \Procedure{Multi-Round}{$x_1, \dots, x_k,t$} 
        \If {$t = 1$}
        \State \textbf{return} the winner of \textsc{Round-Robin} on $x_1, \dots, x_k$.
        \EndIf
        \State Set $\eta_t = \frac{1}{2^t - 1}$.
        \State Partition $x_1$ through $x_k$ into $k^{1-\eta_t}$ sets of size $k^{\eta_t}$. \label{ln:partition}
        \State Run \textsc{Round-Robin} on each set to obtain $k^{1-\eta_t}$ winners. \label{ln:round-robin}
        \State \textbf{return} the winner of \textsc{Multi-Round} on the set of $k^{1-\eta_t}$ winners with $t-1$ rounds. \label{ln:recursive}
        \EndProcedure
    \end{algorithmic}
\end{algorithm}

\begin{figure}[h!]
  \includegraphics[scale=0.5]{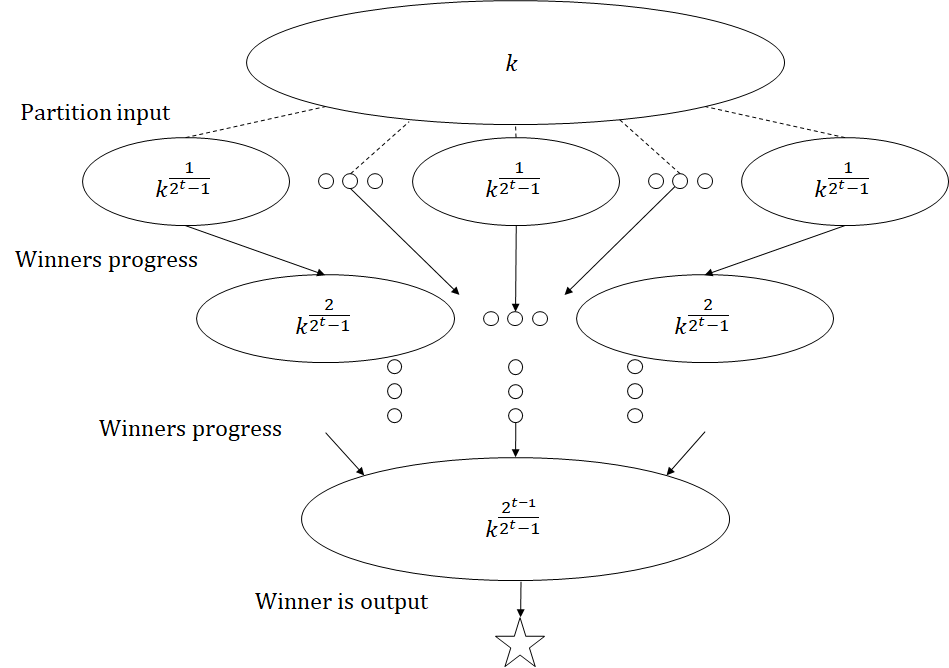}
  \centering
  \caption{An illustration of Algorithm~\ref{alg:t-round}. The input is partitioned into several sets and a round-robin tournament is performed on each. In subsequent rounds, winners are merged into fewer but larger sets, until we have only a single winner.}
\end{figure}

We proceed with proving that this algorithm satisfies the guarantees stated in Lemma~\ref{lem:t-round}.

\begin{proof}
  We prove the guarantees by induction.
  The base case corresponds to $t = 1$.
  As mentioned before, this is exactly equal to Algorithm~\ref{alg:round-robin}, and thus by Claim~\ref{clm:round-robin}, the lemma holds.

  Now, we prove the lemma for a general $t>1$, assuming it holds for $t-1$.
  The number of rounds is trivial: $1$ round is spent performing Lines~\ref{ln:partition} and~\ref{ln:round-robin}, and $t-1$ rounds are spent on the recursive call in Line~\ref{ln:recursive}.
  The approximation is also easy to reason about: the maximum element in the input appears in one of the sets in the partition in Line~\ref{ln:partition}, and therefore the winner of the corresponding set will be a $2$-approximation of the maximum.
  Thus, the set of winners which are fed into the recursive call in Line~\ref{ln:recursive} will have a $2$-approximation of the maximum.
  The inductive hypothesis guarantees that the winner of the recursive call will be a $2(t-1)$-approximation to \emph{this} item, making it a $2t$-approximation to the maximum.
  
  Finally, it remains to reason about the query complexity.
  Comparisons are only performed in Lines~\ref{ln:round-robin} and~\ref{ln:recursive}.
  In the former, we perform the round-robin tournament on $k^{1-\eta_t}$ sets of size $k^{\eta_t}$, so the total number of comparisons is $k^{1-\eta_t} \cdot O(k^{2\eta_t}) = O(k^{1+\eta_t})$.
  In the latter, the recursive call has an input of size $k^{1-\eta_t}$, so by the inductive hypothesis, the number of comparisons done in the recursive call is $O\left(\left(k^{1-\eta_t}\right)^{1 + \frac{1}{2^{t-1}-1}}(t-1)\right)$.
  Substituting in the value $\eta_t = \frac{1}{2^t-1}$, these two terms sum to $O(k^{1+\frac{1}{2^t -1}}t)$, as desired.
\end{proof}

\subsubsection{Bounding the Approximation Factor}
\label{sec:comp:ub:bound}
While the guarantees of Lemma~\ref{lem:t-round} are strong in terms of the round and query complexity, the approximation leaves something to be desired.
We alleviate this issue in a similar way as~\cite{DaskalakisK14}, by running a very simple strategy in parallel to the main method of Algorithm~\ref{alg:t-round}.
The intuition is as follows: if an item with maximum value $x^*$ is never compared with an item with value $x'$ such that $x^* > x' \geq 1$ (i.e., numbers which are $1$-approximations to the maximum), it will never lose a comparison.
If the fraction of such elements is low, then an item with value $x^*$ will make it to the final round, thus guaranteeing that the overall winner will be a $2$-approximation to the maximum.
On the other hand, if the fraction of such elements is high, then we can sample a small number of items such that we select at least one $1$-approximation to $x^*$, and running the round-robin algorithm on this set will guarantee a $3$-approximation to the maximum.

Our method is described more precisely in Algorithm~\ref{alg:better-t-round}, and the guarantees are described in Theorem~\ref{thm:better-t-round}.
\begin{algorithm}
    \caption{Better $t$-Round Algorithm for Maximum Selection}
    \label{alg:better-t-round}
    \hspace*{\algorithmicindent} \textbf{Input:} $K$ items $x_1, \dots, x_k$, number of rounds $t$\\
    \hspace*{\algorithmicindent} \textbf{Output:} Approximate maximum $x_i$ 
    \begin{algorithmic}[1] 
        \Procedure{Better-Multi-Round}{$x_1, \dots, x_k,t$} 
        \State Run \textsc{Multi-Round} on a random permutation of $x_1, \dots x_k$ with $t$ rounds, but halt when $t =1$ and let $L$ be the set of all remaining items. \label{ln:better:mr-call}
        \State Let $H$ be a random subset of $\{x_1, \dots, x_k\}$ of size $O\left(k^\frac{2^{t-1}}{2^{t}-1}\right)$. \label{ln:better:subset}
        \State Run \textsc{Round-Robin} on $L \cup H$ and \textbf{return} the winner. \label{ln:better:rr}
        \EndProcedure
    \end{algorithmic}
\end{algorithm}

\begin{figure}[h!]
  \includegraphics[scale=0.5]{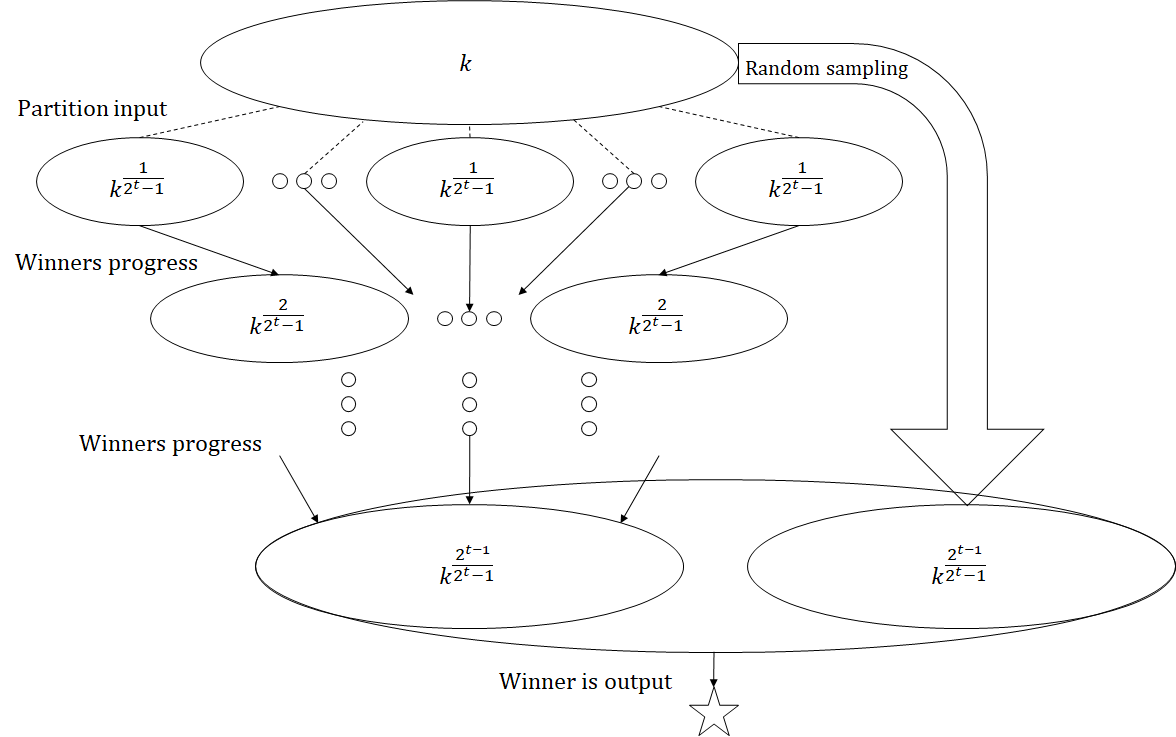}
  \centering
  \caption{An illustration of Algorithm~\ref{alg:better-t-round}. Similar to Algorithm~\ref{alg:t-round}, but in the last round, we perform a round-robin tournament additionally involving a random sample of items from the input.}
\end{figure}

\begin{theorem}
  \label{thm:better-t-round}
  There exists a $t$-round algorithm which, with probability $9/10$, achieves a $3$-approximation in the problem of parallel approximate maximum selection with adversarial comparators.
  The algorithm requires $O(k^{1 + \frac{1}{2^t-1}}t)$ queries.
\end{theorem}

This gives the following corollary for LDP hypothesis selection. 
\begin{corollary}
  \label{cor:better-t-round-ldp}
  There exists a $t$-round algorithm which achieves a $(27+\gamma)$-agnostic approximation factor for locally private hypothesis selection with probability $9/10$, where $\gamma > 0$ is an arbitrarily small constant.
  The sample complexity of the algorithm is $O\left(\frac{k^{1 + \frac{1}{2^t-1}}t \log k}{\ve^2\alpha^2}\right)$.
\end{corollary}
Corollaries~\ref{cor:set-params-select} and~\ref{cor:set-params-ldp} follow from these statements with an appropriate setting of $t$.

To conclude, we prove Theorem~\ref{thm:better-t-round}.

\begin{proof}
  The number of rounds is straightforward to analyze: Line~\ref{ln:better:mr-call} takes $t-1$ rounds (since we stop one round early), and Lines~\ref{ln:better:subset} and~\ref{ln:better:rr} can be done in $1$ last round.

  To analyze the number of comparisons, we require the following claim, which quantifies the number of items that make it to the last round of \textsc{Multi-Round}.
  \begin{claim}
    \label{clm:bound-round}
    $|L| = k^\frac{2^{t-1}}{2^t -1}$.
  \end{claim}
  \begin{proof}
    We recall the $\eta_t$ notation of Algorithm~\ref{alg:t-round}.
    The number of items which begin the first round of the algorithm is clearly $k$.
    Since these are partitioned into $k^{1-\eta_t}$ groups, each producing a single winner which progresses to the next round, we have $k^{1-\eta_t}$ items which begin the second round of the algorithm.
    A similar reasoning implies that the number of items entering the third round of the algorithm is $\left(k^{1-\eta_t}\right)^{1-\eta_{t-1}}$.
    Noting that $|L|$ is the number of items entering the $t$-th (i.e., final) round of the algorithm, the same logic shows that 
    \[
      \log_k |L| = \prod_{i=0}^{t-2} \left(1 - \frac1{2^{t-i}-1}\right) = \prod_{i=0}^{t-2} \left(\frac{2(2^{t-i-1}-1)}{2^{t-i}-1}\right) = \frac{2^{t-1}}{2^t - 1},
    \]
    as desired.
    The latter equality can be seen by a telescoping argument, as the numerators cancel the subsequent denominators.
  \end{proof}
  With this in hand, the number of comparisons is the number of comparisons due to Line~\ref{ln:better:mr-call} (which is $O\left(k^{1 + \frac{1}{2^t-1}}(t-1)\right)$ by the same argument as in the proof of Lemma~\ref{lem:t-round}) plus the number of comparisons due to Line~\ref{ln:better:rr}, which is $O\left((|H| + |L|)^2\right) = O\left(\left(k^\frac{2^{t-1}}{2^{t}-1}\right)^2\right) =  O\left(k^{1 + \frac{1}{2^t-1}}\right)$.
  Combining both of these gives the desired number of comparisons.

  Finally, we justify the accuracy guarantee.
  We split the analysis into two cases, based on the density of items which have value comparable to the maximum.
  Let $\zeta = \max_i x_i$ be the maximum value, let $B = \{ i\ :\ \zeta > x_i \geq \zeta - 1\}$ be the set of items which are $1$-approximations to (but not strictly equal to) the maximum value, and let $\gamma = \frac{|B|}{k}$ be their density.

  First, suppose that $\gamma \geq 1/10k^\frac{2^{t-1}}{2^{t}-1}$.
  If we let the hidden constant in the size of $H$ be $100$ (i.e., $|H| = 100k^\frac{2^{t-1}}{2^{t}-1}$), then Markov's inequality says that at least one item in $H$ will be a $1$-approximation to the maximum value with probability at least $9/10$.
  By the guarantees of \textsc{Round-Robin} (quantified in Claim~\ref{clm:round-robin}), the result of Line~\ref{ln:better:rr} will be a $3$-approximation to the maximum value, as desired.

  On the other hand, suppose that $\gamma \leq 1/10k^\frac{2^{t-1}}{2^{t}-1}$.
  We argue that an item with value $\zeta$ makes it to the final round of \textsc{Multi-Round} and is included in $L$ -- if this happens, then by the guarantees of \textsc{Round-Robin}, the result of Line~\ref{ln:better:rr} will be a $2$-approximation to $\zeta$ and the proof is complete.
  This happens if an item with value $\zeta$ is never compared to any element from $B$ within the first $t-1$ rounds.
  Fix some such item: the probability it is compared with some element from $B$ is upper bounded by the probability that any element of $B$ appears in the same subtree of depth $t-1$ leading up to the final round.
  The number of elements contained in this subtree is $k/|L| = k^\frac{2^{t-1}-1}{2^t -1}$, by Claim~\ref{clm:bound-round}.
  The expected number of items from $B$ in this subtree is bounded as $\gamma \cdot k/|L| \leq \frac{1}{10k^\frac{1}{2^t-1}} \leq 1/10$, and the result follows again from Markov's inequality.
\end{proof}

\subsection{A Lower Bound for Selection via Adversarial Comparators}
\label{sec:comp:lb}

\newcommand{\fa}{\tau}

In this section, we provide a lower bound for adversarial maximum selection with constrained interactivity. In Section~\ref{sec:two-round-lower-bound}, we consider a special case when $t=2$ and prove that any 2-round algorithm requires $\Omega\Paren{k^{\frac{4}{3}}}$ comparisons to find an approximate maximum. In Section~\ref{sec:t-round-lower-bound}, we generalize our result and technique to $t$ rounds and prove that any $t$-round algorithm requires $\Omega\Paren{ \frac { k^{1+ \frac{1}{2^{t}-1}}}{3^t} }$ comparisons. These lower bounds hold even for a non-adaptive adversary.

\subsubsection{A Lower Bound for $2$-round Algorithms}
\label{sec:two-round-lower-bound}

We warm up with a simpler case which illustrates the main ideas, namely, a lower bound for $2$-round algorithms. Specifically, we show that no matter how large the approximation factor $\fa$ is, any 2-round algorithm which solves the parallel approximate maximum selection problem requires $\Omega(k^{\frac{4}{3}})$ comparisons. 

\begin{theorem}
 \label{thm:lower-2-round}
For  any  $\fa > 1$, any $2$-round algorithm which achieves a $\fa$-approximation in the problem of parallel approximate maximum selection with non-adaptive adversarial comparators requires $\Omega(k^{\frac{4}{3}})$ queries.
\end{theorem}

We remark that, since our result is proved in the setting of non-adaptive adversarial comparators, it also automatically holds for adaptive comparators as well.

In our lower bound constructions, we reformulate the parallel approximate maximum selection problem as a game between an adversary and the algorithm. Before the game starts, the adversary commits to a random tournament (i.e., a complete directed graph)\footnote{Note that in this section we use ``tournament'' in the graph theoretic sense.} on $k$ nodes, each identified with one of the $k$ items. We will require that the tournament has, with probability $1$, a single sink node. Then, the algorithm player asks $m_1$ queries to the adversary, each query corresponding to a comparison between items $x_i$ and $x_j$. If the corresponding edge between $x_i$ and $x_j$ in the tournament is directed from $x_i$ to $x_j$, then the adversary answers that $x_j > x_i$, and, otherwise, the adversary answers that $x_i > x_j$. Equivalently, the algorithm asks for the directions of $m_1$ edges, which are revealed by the adversary. Afterwards, the player asks $m_2$ additional queries, based on the information gained from the initial $m_1$ queries, and the adversary answers them according to the directions of edges in the tournament. The game continues in this manner for $t$ rounds, where in round $q$ the algorithm asks $m_q$ queries, possibly dependent on all the query answers so far. After the $t$-th round, the algorithm must declare the ``winner'', i.e., the sink in the tournament.

Note that we can always produce item values so that the query answers are valid for the adversarial comparators model, and the sink node is the unique $\fa$-approximate maximum. Let $C_1, \ldots C_\ell$ be the strongly connected components of the tournament, ordered so that, if $i < j$, then all edges between $C_i$ and $C_j$ are directed from $C_i$ to $C_j$. Then we can set, for example, $x_j = 2i\fa$ for all $x_j$ in $C_i$. This way all queries to two items in the same strongly connected component can be answered arbitrarily, and all queries to items in two different components can be answered according to the direction of edges in the tournament. Moreover, we want to mention two special components. First, since there is a unique sink node $x_{i^*}$, $C_\ell$ must be equal to $\{x_{i^*}\}$, and therefore, $x_{i^*}$ is the unique $\fa$-approximate maximum. Second, in order to ``fool'' the player, the adversary sets $C_{\ell-1} = \{ x_{i^{\prime}}\}$, where all edges incident on $x_{i^\prime}$ are directed towards $x_{i^\prime}$, except the edge from $x_{i^\prime}$ to $x_{i^*}$. Thus, if the algorithm can achieve a $\fa$-approximation in the parallel approximate maximum selection problem, it can identify the sink node in the game above, and especially, distinguish it from $ x_{i^{\prime}}$.

We are now ready to prove Theorem~\ref{thm:lower-2-round}.

\begin{proof}
  We model the problem as the game described above, with $t =2$. By Yao's minimax principle, we can assume, without loss of generality, that the algorithm player makes deterministic choices.  
We start with the construction of the random tournament. From now on, to make the notation more convenient, we will denote nodes/items by their indices, i.e., we will write $i$ rather than $x_i$. Let $U_0$ denote the complete set of the nodes. Firstly, the adversary picks a uniformly random subset $U_1$ of $k^{\frac{2}{3}}$ nodes from $U_0$. Then from the adversary picks two nodes $i^*$ and $i^\prime$ uniformly at random from $U_1$.


Now we describe the directions of the edges of the tournament. For convenience, we define $V_0 \coloneqq U_0 \backslash U_1$ and $V_1 \coloneqq U_1 \backslash \{i^*,i'\}$. All edges incident on $i^*$ are directed towards $i^*$, i.e., $i^*$ is our sink node. All edges incident on $i^\prime$ are directed towards $i^\prime$, except the edge from $i^\prime$ to $i^*$. All edges from $V_0$ to $V_1$ are directed towards the node in $V_1$. Finally, the direction of any edge between two nodes in $V_0$ or two nodes in $V_1$ is chosen uniformly and independently from all other random choices.

Now we switch to the side of the player. As noted above, any algorithm which achieves $\fa$-approximation must correctly identify $i^{*}$ as the sink, with probability higher than $\frac{2}{3}$. Given $m_1 = m_2 = \frac{1}{100} k^\frac{4}{3}$, we want to show that any algorithm which asks $m = m_1 + m_2$ queries can not find $i^*$ with this probability. In the first round, the player asks $m_1 =  \frac{1}{100} k^\frac{4}{3} $ number of queries. We use $e_j = \{\alpha_j, \beta_j\}$ to denote the $j$-th query, where $j \in [m_1]$. Let $S$ denote the set of the nodes in $U_1$ which have ever competed with some other nodes from $U_1$, i.e., $S = \{ i_1 \in U_1 : \exists i_2 \in U_1, \exists j \in [m_1], e_j = \{i_1,i_2\}\}$.
Now we want to show that the following two ``bad'' events happen with a small probability:

$$ A_1 = \{ i^{\prime} \in S ~ \cup ~ i^{*} \in S \}, ~~~~ A_2 = \{  | V_1 \cap S |  \ge \frac1{2} \cdot k^\frac{2}{3}\}.$$

We bound the probability of event $A_1$ and $A_2$, respectively. For the rest of the proof, we will assume that $k$ is a large enough constant. 
By a union bound, 
$$\pr{A_1} \le \pr{ i^{\prime} \in S } + \pr{i^{*} \in S}  \le 2 \cdot m_1 \cdot  \frac{k^{\frac{2}{3}}-1}{\binom{k}{2}} \le 0.05.$$

With respect to $A_2$, let $e_j = \{\alpha_j, \beta_j\}$, where $j \in [m_1]$. We note that $| V_1 \cap S | \le 2\cdot \sum_{j\in [m_1]} \id (\alpha_j \in V_1, \beta_j \in V_1)$, where $\forall j$, $ \E \Paren{ \id (\alpha_j \in V_1, \beta_j \in V_1)} = \binom {k^{\frac{2}{3} }} {2} \backslash \binom{k}{2}  = \frac{k^{\frac{2}{3}}-1 }{k-1} \cdot k^{-\frac13}$. Furthermore, $\forall j_1 \neq j_2$, $\id (\alpha_{j_1} \in V_1, \beta_{j_1} \in V_1)$ and $\id (\alpha_{j_2} \in V_1, \beta_{j_2} \in V_1)$ are negatively correlated.

Therefore,
$$ \E \Paren{\sum_{j\in [m_1]} \id (\alpha_j \in V_1, \beta_j \in V_1)}  = m_1 \cdot \frac{k^{\frac{2}{3}}-1 }{k-1} \cdot k^{-\frac13} = \frac1{100}\cdot \frac{k}{k-1}\cdot (k^{\frac{2}{3}}-1),$$
$$ \Var \Paren{\sum_{j\in [m_1]} \id (\alpha_j \in V_1, \beta_j \in V_1)}  \le m_1 \cdot \frac{k^{\frac{2}{3}}-1 }{k-1} \cdot k^{-\frac13}\cdot \Paren{1 - \frac{k^{\frac{2}{3}}-1 }{k-1} \cdot k^{-\frac13} } < \frac1{100}\cdot \frac{k}{k-1}\cdot (k^{\frac{2}{3}}-1).$$

By Chebyshev's inequality,  
$$ \pr {A_2} \le \pr {\sum_{j\in [m_1]} \id (\alpha_j \in V_1, \beta_j \in V_1) \ge \frac1{4} \cdot k^\frac{2}{3} } \le k^{-\frac{2}{3}} \le 0.05, $$
where in the last inequality, we assume $k\ge 100$.

Now we move to the second round. From now on, we condition on neither $A_1$ nor $A_2$ holding, which happens with probability at least $0.9$. Then, conditional on $A_1$ and on the answers to the first $m_1$ queries, the pair $\{i^*, i^\prime\}$ is distributed uniformly in the set $R = \{i^*, i^\prime\} \cup (V_1 \setminus S)$. Moreover, if the algorithm does not query $\{i^*, i^\prime\}$ in the second round, then $i^*$ and $i^\prime$ will have the same distribution conditional on all $m$ queries, and the algorithm will not be able to identify $i^*$ with probability higher than $0.5$. Then, conditional on $A_1$, $A_2$, and the queries from the first round, the probability that the algorithm queries $\{i^*, i^\prime \}$ in the second round is at most
\[
  m_2 \cdot \frac1{\binom{|R|}{2}} \le \frac{k^{4/3}}{50 \cdot \frac1{2}k^{2/3}\cdot (\frac1{2} k^{2/3}-1)} \le 0.1.
\]
Therefore, the success rate of any deterministic $2$-round algorithm making at most $\frac{k^{2/3}}{100}$ queries is at most $0.1 + 0.1 + 0.5 < \frac{9}{10}$. As already noted, by Yao's minimax principle this also implies the result for randomized algorithms.


\end{proof}

\subsubsection{A Lower Bound for $t$-round Algorithms}
\label{sec:t-round-lower-bound}

In this section, we extend our 2-round lower bound to $t$ rounds. Specifically, we want to prove the following theorem.

\begin{theorem}
\label{thm:lower-t-round}
For  any  $\fa > 1$, any $t$-round algorithm which achieves  $\fa$-approximation in the problem of parallel approximate maximum selection with non-adaptive adversarial comparators requires $\Omega\Paren{ \frac { k^{1+ \frac{1}{2^{t}-1}}}{3^t} }$ queries.
\end{theorem}

We continue to model the problem as the game described in the previous subsection, but now with general $t$. We start with the construction of the random tournament, where a similar hierarchical structure to the 2-round construction is adopted. In the structure in Section~\ref{sec:two-round-lower-bound}, we can view node $i^*$ and $i^{\prime}$ as layer $2$, nodes in set $V_1$ as layer $1$, and all the other nodes as layer 0. We have thus designed a $3$-layer hierarchical structure in the proof of the 2-round lower bound, where edges are directed from lower to higher layers, and edges in the same layer are directed randomly. In this section, we generalize this construction to the following $(t+1)$-layer hierarchical structure, which we denote as $(k,t)$-construction.

Let $U_0$ denote the complete set of the nodes. In the first round, the adversary uniformly at random picks $k^{\frac{2^t-2}{2^t-1}}$ different nodes from $U_0$, which are denoted as $U_1$; etc.; in the $q$-th round, the adversary uniformly randomly picks $k^{\frac{2^t -2 ^q}{2^t-1} }$ from $U_{q-1}$, denoted as $U_q$, where $q \in [t-1]$. Finally, the adversary uniformly at random picks two nodes from $U_{t-1}$, denoted as $i^*$ and $i^{\prime}$, respectively, and we let $U_t = \{ i^*, i^{\prime}\}$ for the purpose of consistency.
For convenience, we define $V_0 = U_0 \backslash U_1$, $\cdots$, $V_q = U_q \backslash U_{q+1}$, where $0\le q \le t-1$, and $V_t = U_t = \{ i^*, i^{\prime}\}$. For $i < j$, we direct all edges from $V_i$ to $V_j$; for $q \neq t$,  edges between two nodes in $V_q$ are given a uniformly random direction; finally, the edge between $i^*$ and $i^\prime$ is directed towards $i^*$. Thus, $i^*$ is the unique sink in the graph.

The following is the core lemma in this section.

\begin{lemma}
\label{lem:main-t-lower}
Given a $(k,t)$-construction, and $ \forall \gamma<1$, every
deterministic $t$-round algorithm which finds $i^*$ with
probability higher than $\Paren{ \frac12+ \frac{\gamma}{100} \cdot
  3^{t}}$ requires $\Omega\left(\gamma\Paren{k^{1+ \frac{1}{2^{t}-1}}}\right)$ queries.
\end{lemma}

It is not hard to show that Theorem~\ref{thm:lower-t-round} can be
viewed as a corollary of the lemma, since given the random
$(k,t)$-construction, by setting $\gamma=\frac1{3^t}$, the lemma tells
that every $t$-round algorithm which finds $i^{*}$ with constant
probability makes at least
$\Omega\Paren{ \frac { k^{1+ \frac{1}{2^{t}-1}}}{3^t} }$ queries, and
any algorithm which achieves $\tau$-approximation should find $i^{*}$
with constant probability. Finally, by Yao's minimax principle, this
also holds for randomized algorithms. Therefore, our remaining task is to prove
Lemma~\ref{lem:main-t-lower}.

\begin{proof}
We prove the lemma by induction. Throughout the proof we assume that
$k$ is large enough with respect to $t$ and $\frac1\gamma$. We will
assume that the algorithm makes at most
$\frac{\gamma}{100}\Paren{k^{1+ \frac{1}{2^{t}-1}}}$ queries, and show
inductively that it succeeds in identifying $i^*$ with probability at
most $\frac12+ \frac{\gamma}{100} \cdot 3^{t}$.

For the base case when $t=2$, the lemma holds from the argument in the previous section.
For the inductive step, let $t$ be any integer where $t\ge 3$. Recall
that the number of queries asked by the algorithm in the first round
is $m_1 \le \frac{\gamma}{100} k^{1+ \frac{1}{2^{t}-1 }} $. We use
$e_j = (\alpha_j, \beta_j)$ to denote the $j$-th query, where $j \in
[m_1]$. By analogy with the $2$-round proof, let $S$ denote the set of nodes in $U_1$ which have ever competed with some other nodes from $U_1$, i.e., $S = \{ i_1 \in U_1 : \exists i_2 \in U_1, \exists j \in [m_1], e_j = (i_1,i_2)\}$. Now we want to show that the following $t$ ``bad'' events happen with a small probability:

$$ \forall q \in [t-1], A_q = \{  | V_q \cap S |  \ge \frac{1}{10} \cdot k^{\frac{2^t -2 ^q}{2^t-1} }\}, ~~~~ A_{t} = \{ i^{\prime} \in S ~ \cup ~ i^{*} \in S \}.$$

We bound the probability of event $A_t$ first. By a union bound, 
$$\pr{A_t} \le \pr{ i^{\prime} \in S } + \pr{i^{*} \in S}  \le 2 \cdot m_1 \cdot  \frac{k^{\frac{2^t -2}{2^t-1} }-1}{\binom{k}{2}} \le 0.05 \gamma.$$

With respect to $A_q, q \in [t-1]$, let $e_j = (\alpha_j, \beta_j)$, where $j \in [m_1]$. We note that $| V_q \cap S | \le 2\cdot \sum_{j\in [m_1]} \id (\alpha_j \in V_1, \beta_j \in V_q)$, where $\forall j$, $ \E \Paren{ \id (\alpha_j \in V_1, \beta_j \in V_q)} =  \frac{ k^{\frac{2^t-2}{2^t-1}} \cdot  k^{\frac{2^t-2^q}{2^t-1}}} { \binom{k}{2} }$, which is roughly $k^{-\frac{2^q}{2^t-1}}$. Furthermore, $\forall j_1 \neq j_2$, $\id (\alpha_{j_1} \in V_1, \beta_{j_1} \in V_q)$ and $\id (\alpha_{j_2} \in V_1, \beta_{j_2} \in V_q)$ are negatively correlated.
Therefore,
$$ \E \Paren{\sum_{j\in [m_1]} \id (\alpha_j \in V_1, \beta_j \in V_q)}  = m_1 \cdot \frac{ k^{\frac{2^t-2}{2^t-1}} \cdot  k^{\frac{2^t-2^q}{2^t-1}}} { \binom{k}{2} }\le \frac{1}{50}\cdot k^{\frac{2^t -2 ^q}{2^t-1} },$$
$$ \Var \Paren{\sum_{j\in [m_1]} \id (\alpha_j \in V_1, \beta_j \in V_q)} \le m_1\cdot \frac{ k^{\frac{2^t-2}{2^t-1}} \cdot  k^{\frac{2^t-2^q}{2^t-1}}} { \binom{k}{2} }\cdot \Paren{1-\frac{ k^{\frac{2^t-2}{2^t-1}} \cdot  k^{\frac{2^t-2^q}{2^t-1}}} { \binom{k}{2} }}  \le\frac{1}{50}\cdot k^{\frac{2^t -2 ^q}{2^t-1} }.$$

By Chebyshev's inequality,  
$$ \pr {A_q} \le \pr {\sum_{j\in [m_1]} \id (\alpha_j \in V_1, \beta_j
  \in V_q) \ge \frac{1}{20} \cdot k^{\frac{2^t -2 ^q}{2^t-1}} } \le 25
k^{- \frac{2^{t-1}}{2^t-1}} \le \frac{0.05\gamma}{t}, $$
where in the last inequality, we assume $k \ge \frac{Ct^2}{\gamma^2}$
for a large enough constant $C$.

From now on, we condition on none of the bad events $A_1, \ldots, A_t$
holding, which happens with probability at least $1 - 0.1\gamma$. We
also condition on the answers to the first $m_1$ queries. We would
like to say that the conditional distribution on the graph induced on
$U_1 \setminus S$ is identical to that of a
$(k^\prime, t-1)$-construction for $k^\prime = U_1 \setminus
S$. However, because of the random choice of $S$, the sizes of
$U_q\setminus S$ are not exactly as prescribed in the definition of a
$(k^\prime, t-1)$ construction. In order to finish the induction, we
consider the following process. For
$ k^{\prime} = \frac12 k^{\frac{2^t -2}{2^t-1} }$, we first 
denote $V^\prime_t = \{i^*, i^{\prime}\}$; then, we uniformly at
random draw $(k^{\prime})^{\frac{2^{t-2}}{2^{t-1}-1} } -2 $ nodes from
from $V_{t-1} \backslash S$,  and denote them as $V^{\prime}_{t-1}$; from
$V_q \backslash S$, $q \in [t-1]$, we uniformly at random draw
$(k^{\prime})^{\frac{2^{t-1} -2^{q-1}}{2^{t-1}-1} } -
|V^{\prime}_{q+1}| < \Paren{\frac12}^{\frac{2^{t-1}
    -2^{q-1}}{2^{t-1}-1}} \cdot k^{\frac{2^t -2^q}{2^t-1}} <
\frac{3}{4}k^{\frac{2^t -2^q}{2^t-1}}$ nodes, and denote them as
$V^{\prime}_q$. Conditonal on the bad events not holding, and on the
query answers from the first round, the subgraph induced on the nodes from $V^{\prime}_1$,
$V^{\prime}_2$, $\cdots$, and $V^{\prime}_t$, is distributed
identically to a $(k^{\prime},t-1)$ construction. Clearly, for the
algorithm to determine the sink $i^*$ in the full tournament, it must
also determine it in this subgraph. Ignoring queries in rounds
$2, \ldots, t$ to edges not in the subgraph, the algorithm is allowed to
ask at most 
$m = \frac{\gamma}{100}\Paren{k^{1+ \frac{1}{2^{t}-1}}} \le
\frac{\gamma}{100} \cdot 2.7 \cdot \Paren{ \Paren{k^{\prime}}^{1+
    \frac{1}{2^{t-1}-1}}}$ queries, and, by the inductive assumption, any
$(t-1)$-round algorithm can find $i^{*}$ with probability at most
$\frac12+ \frac{3^{t-1}}{100}\cdot 2.7\gamma $. Finally, by a union bound, the
probability of success of the $t$-round algorithm is at most
$\frac12+ \frac{3^{t-1}}{100}\cdot 2.7\gamma + 0.1 \gamma \le \frac12+ \frac{ 3^t}{100}
\gamma$. This finishes the inductive step.
\end{proof}

\section*{Acknowledgments}
The main algorithms of Section~\ref{sec:comparators} are based off of (non-optimized) ideas from a previous version of~\cite{DaskalakisK14} (see v1 on arXiv), and GK would like to thank Constantinos Daskalakis for his guidance and ideas in that project.
GK would also like to thank Sepehr Assadi and Jieming Mao for helpful discussions related to the literature on parallel selection.

\appendix

\bibliographystyle{alpha}
\bibliography{biblio}

\newcommand{\etalchar}[1]{$^{#1}$}
\begin{thebibliography}{CAMTM20}

\bibitem[AA88a]{AlonA88a}
Noga Alon and Yossi Azar.
\newblock The average complexity of deterministic and randomized parallel
  comparison-sorting algorithms.
\newblock {\em SIAM Journal on Computing}, 17(6):1178--1192, 1988.

\bibitem[AA88b]{AlonA88b}
Noga Alon and Yossi Azar.
\newblock Sorting, approximate sorting, and searching in rounds.
\newblock {\em SIAM Journal on Discrete Mathematics}, 1(3):269--280, 1988.

\bibitem[AAAK17]{AgarwalAAK17}
Arpit Agarwal, Shivani Agarwal, Sepehr Assadi, and Sanjeev Khanna.
\newblock Learning with limited rounds of adaptivity: Coin tossing, multi-armed
  bandits, and ranking from pairwise comparisons.
\newblock In {\em Proceedings of the 30th Annual Conference on Learning
  Theory}, COLT '17, pages 39--75, 2017.

\bibitem[AAV86]{AlonAV86}
Noga Alon, Yossi Azar, and Uzi Vishkin.
\newblock Tight complexity bounds for parallel comparison sorting.
\newblock In {\em Proceedings of the 27th Annual IEEE Symposium on Foundations
  of Computer Science}, FOCS '86, pages 502--510, Washington, DC, USA, 1986.
  IEEE Computer Society.

\bibitem[ACFT19]{AcharyaCFT19}
Jayadev Acharya, Cl{\'e}ment~L. Canonne, Cody Freitag, and Himanshu Tyagi.
\newblock Test without trust: Optimal locally private distribution testing.
\newblock In {\em Proceedings of the 22nd International Conference on
  Artificial Intelligence and Statistics}, AISTATS '19, pages 2067--2076. JMLR,
  Inc., 2019.

\bibitem[ACT19]{AcharyaCT19}
Jayadev Acharya, Cl{\'e}ment~L. Canonne, and Himanshu Tyagi.
\newblock Inference under information constraints: Lower bounds from chi-square
  contraction.
\newblock In {\em Proceedings of the 32nd Annual Conference on Learning
  Theory}, COLT '19, pages 1--15, 2019.

\bibitem[ADKR19]{AliakbarpourDKR19}
Maryam Aliakbarpour, Ilias Diakonikolas, Daniel~M. Kane, and Ronitt Rubinfeld.
\newblock Private testing of distributions via sample permutations.
\newblock In {\em Advances in Neural Information Processing Systems 32},
  NeurIPS '19, pages 10877--10888. Curran Associates, Inc., 2019.

\bibitem[ADR18]{AliakbarpourDR18}
Maryam Aliakbarpour, Ilias Diakonikolas, and Ronitt Rubinfeld.
\newblock Differentially private identity and closeness testing of discrete
  distributions.
\newblock In {\em Proceedings of the 35th International Conference on Machine
  Learning}, ICML '18, pages 169--178. JMLR, Inc., 2018.

\bibitem[AFHN09]{AjtaiFHN09}
Mikl{\'o}s Ajtai, Vitaly Feldman, Avinatan Hassidim, and Jelani Nelson.
\newblock Sorting and selection with imprecise comparisons.
\newblock In {\em Proceedings of the 36th International Colloquium on Automata,
  Languages, and Programming}, ICALP '09, pages 37--48, 2009.

\bibitem[AFJ{\etalchar{+}}18]{AcharyaFJOS18}
Jayadev Acharya, Moein Falahatgar, Ashkan Jafarpour, Alon Orlitsky, and
  Ananda~Theertha Suresh.
\newblock Maximum selection and sorting with adversarial comparators.
\newblock {\em Journal of Machine Learning Research}, 19(1):2427--2457, 2018.

\bibitem[AJM19]{AminJM19}
Kareem Amin, Matthew Joseph, and Jieming Mao.
\newblock Pan-private uniformity testing.
\newblock {\em arXiv preprint arXiv:1911.01452}, 2019.

\bibitem[AJOS14]{AcharyaJOS14b}
Jayadev Acharya, Ashkan Jafarpour, Alon Orlitsky, and Ananda~Theertha Suresh.
\newblock Sorting with adversarial comparators and application to density
  estimation.
\newblock In {\em Proceedings of the 2014 IEEE International Symposium on
  Information Theory}, ISIT '14, pages 1682--1686, Washington, DC, USA, 2014.
  IEEE Computer Society.

\bibitem[AKS83]{AjtaiKS83}
Mikl{\'o}s Ajtai, J{\'a}nos Koml{\'o}s, and Endre Szemer{\'e}di.
\newblock An $o(n \log n)$ sorting network.
\newblock In {\em Proceedings of the 15th Annual ACM Symposium on the Theory of
  Computing}, STOC '83, pages 1--9, New York, NY, USA, 1983. ACM.

\bibitem[AKSZ18]{AcharyaKSZ18}
Jayadev Acharya, Gautam Kamath, Ziteng Sun, and Huanyu Zhang.
\newblock Inspectre: Privately estimating the unseen.
\newblock In {\em Proceedings of the 35th International Conference on Machine
  Learning}, ICML '18, pages 30--39. JMLR, Inc., 2018.

\bibitem[Alo85]{Alon85}
Noga Alon.
\newblock Expanders, sorting in rounds and superconcentrators of limited depth.
\newblock In {\em Proceedings of the 17th Annual ACM Symposium on the Theory of
  Computing}, STOC '85, pages 98--102, New York, NY, USA, 1985. ACM.

\bibitem[AP90]{AzarP90}
Yossi Azar and Nicholas Pippenger.
\newblock Parallel selection.
\newblock {\em Discrete Applied Mathematics}, 27(1-2):49--58, 1990.

\bibitem[AS18]{AwanS18}
Jordan Awan and Aleksandra Slavkovi{\'c}.
\newblock Differentially private uniformly most powerful tests for binomial
  data.
\newblock In {\em Advances in Neural Information Processing Systems 31},
  NeurIPS '18, pages 4208--4218. Curran Associates, Inc., 2018.

\bibitem[ASZ18]{AcharyaSZ18}
Jayadev Acharya, Ziteng Sun, and Huanyu Zhang.
\newblock Differentially private testing of identity and closeness of discrete
  distributions.
\newblock In {\em Advances in Neural Information Processing Systems 31},
  NeurIPS '18, pages 6878--6891. Curran Associates, Inc., 2018.

\bibitem[AV87]{AzarV87}
Yossi Azar and Uzi Vishkin.
\newblock Tight comparison bounds on the complexity of parallel sorting.
\newblock {\em SIAM Journal on Computing}, 16(3):458--464, 1987.

\bibitem[BB90]{BollobasB90}
B{\'e}la Bollob{\'a}s and Graham Brightwell.
\newblock Parallel selection with high probability.
\newblock {\em SIAM Journal on Discrete Mathematics}, 3(1):21--31, 1990.

\bibitem[BGM{\etalchar{+}}16]{BravermenGMNW}
Mark Braverman, Ankit Garg, Tengyu Ma, Huy~L. Nguyen, and David~P. Woodruff.
\newblock Communication lower bounds for statistical estimation problems via a
  distributed data processing inequality.
\newblock In {\em Proceedings of the 48th Annual {ACM} {SIGACT} Symposium on
  Theory of Computing, {STOC} 2016, Cambridge, MA, USA, June 18-21, 2016},
  pages 1011--1020, 2016.

\bibitem[BH85]{BollobasH85}
B{\'e}la Bollob{\'a}s and Pavol Hell.
\newblock Sorting and graphs.
\newblock In Ivan Rival, editor, {\em Graphs and Order: The Role of Graphs in
  the Theory of Ordered Sets and Its Applications}, pages 169--184. Springer,
  1985.

\bibitem[BKM19]{BousquetKM19}
Olivier Bousquet, Daniel~M. Kane, and Shay Moran.
\newblock The optimal approximation factor in density estimation.
\newblock In {\em Proceedings of the 32nd Annual Conference on Learning
  Theory}, COLT '19, pages 318--341, 2019.

\bibitem[BKSW19]{BunKSW19}
Mark Bun, Gautam Kamath, Thomas Steinke, and Zhiwei~Steven Wu.
\newblock Private hypothesis selection.
\newblock In {\em Advances in Neural Information Processing Systems 32},
  NeurIPS '19, pages 156--167. Curran Associates, Inc., 2019.

\bibitem[BMP19]{BravermanMP19}
Mark Braverman, Jieming Mao, and Yuval Peres.
\newblock Sorted top-k in rounds.
\newblock In {\em Proceedings of the 32nd Annual Conference on Learning
  Theory}, COLT '19, pages 342--382, 2019.

\bibitem[BMW16]{BravermanMW16}
Mark Braverman, Jieming Mao, and S.~Matthew Weinberg.
\newblock Parallel algorithms for select and partition with noisy comparisons.
\newblock In {\em Proceedings of the 48th Annual ACM Symposium on the Theory of
  Computing}, STOC '16, pages 851--862, New York, NY, USA, 2016. ACM.

\bibitem[BN14]{BrennerN14}
Hai Brenner and Kobbi Nissim.
\newblock Impossibility of differentially private universally optimal
  mechanisms.
\newblock {\em SIAM Journal on Computing}, 43(5):1513--1540, 2014.

\bibitem[BT83]{BollobasT83}
B{\'e}la Bollob{\'a}s and Andrew Thomason.
\newblock Parallel sorting.
\newblock {\em Discrete Applied Mathematics}, 6(1):1--11, 1983.

\bibitem[CAMTM20]{CohenMM20}
Vincent Cohen-Addad, Frederik Mallmann-Trenn, and Claire Mathieu.
\newblock Instance-optimality in the noisy value-and comparison-model* accept,
  accept, strong accept: Which papers get in?
\newblock In {\em Proceedings of the 31st Annual ACM-SIAM Symposium on Discrete
  Algorithms}, SODA '20, pages 2124--2143, Philadelphia, PA, USA, 2020. SIAM.

\bibitem[CBRG18]{CampbellBRG18}
Zachary Campbell, Andrew Bray, Anna Ritz, and Adam Groce.
\newblock Differentially private {ANOVA} testing.
\newblock In {\em Proceedings of the 2018 International Conference on Data
  Intelligence and Security}, ICDIS '18, pages 281--285, Washington, DC, USA,
  2018. IEEE Computer Society.

\bibitem[CDK17]{CaiDK17}
Bryan Cai, Constantinos Daskalakis, and Gautam Kamath.
\newblock Priv'it: Private and sample efficient identity testing.
\newblock In {\em Proceedings of the 34th International Conference on Machine
  Learning}, ICML '17, pages 635--644. JMLR, Inc., 2017.

\bibitem[CKM{\etalchar{+}}19a]{CanonneKMSU19}
Cl{\'e}ment~L. Canonne, Gautam Kamath, Audra McMillan, Adam Smith, and Jonathan
  Ullman.
\newblock The structure of optimal private tests for simple hypotheses.
\newblock In {\em Proceedings of the 51st Annual ACM Symposium on the Theory of
  Computing}, STOC '19, New York, NY, USA, 2019. ACM.

\bibitem[CKM{\etalchar{+}}19b]{CanonneKMUZ19}
Cl\'ement~L. Canonne, Gautam Kamath, Audra McMillan, Jonathan Ullman, and Lydia
  Zakynthinou.
\newblock Private identity testing for high-dimensional distributions.
\newblock {\em arXiv preprint arXiv:1905.11947}, 2019.

\bibitem[CKS{\etalchar{+}}19]{CouchKSBG19}
Simon Couch, Zeki Kazan, Kaiyan Shi, Andrew Bray, and Adam Groce.
\newblock Differentially private nonparametric hypothesis testing.
\newblock In {\em Proceedings of the 2019 ACM Conference on Computer and
  Communications Security}, CCS '19, New York, NY, USA, 2019. ACM.

\bibitem[CSS11]{chan2011private}
T-H~Hubert Chan, Elaine Shi, and Dawn Song.
\newblock Private and continual release of statistics.
\newblock {\em ACM Transactions on Information and System Security (TISSEC)},
  14(3):1--24, 2011.

\bibitem[DDS12]{DaskalakisDS12b}
Constantinos Daskalakis, Ilias Diakonikolas, and Rocco~A. Servedio.
\newblock Learning {P}oisson binomial distributions.
\newblock In {\em Proceedings of the 44th Annual ACM Symposium on the Theory of
  Computing}, STOC '12, pages 709--728, New York, NY, USA, 2012. ACM.

\bibitem[DF19]{DanielyF19}
Amit Daniely and Vitaly Feldman.
\newblock Locally private learning without interaction requires separation.
\newblock In {\em Advances in Neural Information Processing Systems 32},
  NeurIPS '19, pages 14975--14986. Curran Associates, Inc., 2019.

\bibitem[{Dif}17]{AppleDP17}
{Differential Privacy Team, Apple}.
\newblock Learning with privacy at scale.
\newblock
  \url{https://machinelearning.apple.com/docs/learning-with-privacy-at-scale/appledifferentialprivacysystem.pdf},
  December 2017.

\bibitem[DJW13]{DuchiJW13}
John~C. Duchi, Michael~I. Jordan, and Martin~J. Wainwright.
\newblock Local privacy and statistical minimax rates.
\newblock In {\em Proceedings of the 54th Annual IEEE Symposium on Foundations
  of Computer Science}, FOCS '13, pages 429--438, Washington, DC, USA, 2013.
  IEEE Computer Society.

\bibitem[DJW17]{DuchiJW17}
John~C. Duchi, Michael~I. Jordan, and Martin~J. Wainwright.
\newblock Minimax optimal procedures for locally private estimation.
\newblock {\em Journal of the American Statistical Association}, 2017.

\bibitem[DK14]{DaskalakisK14}
Constantinos Daskalakis and Gautam Kamath.
\newblock Faster and sample near-optimal algorithms for proper learning
  mixtures of {G}aussians.
\newblock In {\em Proceedings of the 27th Annual Conference on Learning
  Theory}, COLT '14, pages 1183--1213, 2014.

\bibitem[DKK{\etalchar{+}}16]{DiakonikolasKKLMS16}
Ilias Diakonikolas, Gautam Kamath, Daniel~M. Kane, Jerry Li, Ankur Moitra, and
  Alistair Stewart.
\newblock Robust estimators in high dimensions without the computational
  intractability.
\newblock In {\em Proceedings of the 57th Annual IEEE Symposium on Foundations
  of Computer Science}, FOCS '16, pages 655--664, Washington, DC, USA, 2016.
  IEEE Computer Society.

\bibitem[DKY17]{DingKY17}
Bolin Ding, Janardhan Kulkarni, and Sergey Yekhanin.
\newblock Collecting telemetry data privately.
\newblock In {\em Advances in Neural Information Processing Systems 30}, NIPS
  '17, pages 3571--3580. Curran Associates, Inc., 2017.

\bibitem[DL96]{DevroyeL96}
Luc Devroye and G\'abor Lugosi.
\newblock A universally acceptable smoothing factor for kernel density
  estimation.
\newblock {\em The Annals of Statistics}, 24(6):2499--2512, 1996.

\bibitem[DL97]{DevroyeL97}
Luc Devroye and G\'abor Lugosi.
\newblock Nonasymptotic universal smoothing factors, kernel complexity and
  {Y}atracos classes.
\newblock {\em The Annals of Statistics}, 25(6):2626--2637, 1997.

\bibitem[DL01]{DevroyeL01}
Luc Devroye and G\'abor Lugosi.
\newblock {\em Combinatorial methods in density estimation}.
\newblock Springer, 2001.

\bibitem[DMNS06]{DworkMNS06}
Cynthia Dwork, Frank McSherry, Kobbi Nissim, and Adam Smith.
\newblock Calibrating noise to sensitivity in private data analysis.
\newblock In {\em Proceedings of the 3rd Conference on Theory of Cryptography},
  TCC '06, pages 265--284, Berlin, Heidelberg, 2006. Springer.

\bibitem[DMR18]{DevroyeMR18b}
Luc Devroye, Abbas Mehrabian, and Tommy Reddad.
\newblock The total variation distance between high-dimensional {G}aussians.
\newblock {\em arXiv preprint arXiv:1810.08693}, 2018.

\bibitem[DR14]{DworkR14}
Cynthia Dwork and Aaron Roth.
\newblock The algorithmic foundations of differential privacy.
\newblock {\em Foundations and Trends{\textregistered} in Machine Learning},
  9(3--4):211--407, 2014.

\bibitem[DR19]{DuchiR19}
John Duchi and Ryan Rogers.
\newblock Lower bounds for locally private estimation via communication
  complexity.
\newblock In {\em Proceedings of the 32nd Annual Conference on Learning
  Theory}, COLT '19, pages 1161--1191, 2019.

\bibitem[EGS03]{EvfimievskiGS03}
Alexandre Evfimievski, Johannes Gehrke, and Ramakrishnan Srikant.
\newblock Limiting privacy breaches in privacy preserving data mining.
\newblock In {\em Proceedings of the 22nd ACM SIGMOD-SIGACT-SIGART Symposium on
  Principles of Database Systems}, PODS '03, pages 211--222, New York, NY, USA,
  2003. ACM.

\bibitem[EPK14]{ErlingssonPK14}
{\'U}lfar Erlingsson, Vasyl Pihur, and Aleksandra Korolova.
\newblock {RAPPOR}: Randomized aggregatable privacy-preserving ordinal
  response.
\newblock In {\em Proceedings of the 2014 ACM Conference on Computer and
  Communications Security}, CCS '14, pages 1054--1067, New York, NY, USA, 2014.
  ACM.

\bibitem[FRPU94]{FeigeRPU94}
Uriel Feige, Prabhakar Raghavan, David Peleg, and Eli Upfal.
\newblock Computing with noisy information.
\newblock {\em SIAM Journal on Computing}, 23(5):1001--1018, 1994.

\bibitem[GLRV16]{GaboardiLRV16}
Marco Gaboardi, Hyun{-}Woo Lim, Ryan~M. Rogers, and Salil~P. Vadhan.
\newblock Differentially private chi-squared hypothesis testing: Goodness of
  fit and independence testing.
\newblock In {\em Proceedings of the 33rd International Conference on Machine
  Learning}, ICML '16, pages 1395--1403. JMLR, Inc., 2016.

\bibitem[GR18]{GaboardiR18}
Marco Gaboardi and Ryan Rogers.
\newblock Local private hypothesis testing: Chi-square tests.
\newblock In {\em Proceedings of the 35th International Conference on Machine
  Learning}, ICML '18, pages 1626--1635. JMLR, Inc., 2018.

\bibitem[HH81]{HaggkvistH81}
Roland H{\"a}ggkvist and Pavol Hell.
\newblock Parallel sorting with constant time for comparisons.
\newblock {\em SIAM Journal on Computing}, 10(3):465--472, 1981.

\bibitem[Hoe94]{hoeffding1994probability}
Wassily Hoeffding.
\newblock Probability inequalities for sums of bounded random variables.
\newblock In {\em The Collected Works of Wassily Hoeffding}, pages 409--426.
  Springer, 1994.

\bibitem[JMNR19]{JosephMNR19}
Matthew Joseph, Jieming Mao, Seth Neel, and Aaron Roth.
\newblock The role of interactivity in local differential privacy.
\newblock In {\em Proceedings of the 60th Annual IEEE Symposium on Foundations
  of Computer Science}, FOCS '19, pages 94--105, Washington, DC, USA, 2019.
  IEEE Computer Society.

\bibitem[JMR20]{JosephMR20}
Matthew Joseph, Jieming Mao, and Aaron Roth.
\newblock Exponential separations in local differential privacy through
  communication complexity.
\newblock In {\em Proceedings of the 31st Annual ACM-SIAM Symposium on Discrete
  Algorithms}, SODA '20, pages 515--527, Philadelphia, PA, USA, 2020. SIAM.

\bibitem[Kea98]{Kearns98}
Michael~J. Kearns.
\newblock Efficient noise-tolerant learning from statistical queries.
\newblock {\em J. {ACM}}, 45(6):983--1006, 1998.

\bibitem[KLN{\etalchar{+}}11]{KasiviswanathanLNRS11}
Shiva~Prasad Kasiviswanathan, Homin~K. Lee, Kobbi Nissim, Sofya Raskhodnikova,
  and Adam Smith.
\newblock What can we learn privately?
\newblock {\em SIAM Journal on Computing}, 40(3):793--826, 2011.

\bibitem[KLSU19]{KamathLSU19}
Gautam Kamath, Jerry Li, Vikrant Singhal, and Jonathan Ullman.
\newblock Privately learning high-dimensional distributions.
\newblock In {\em Proceedings of the 32nd Annual Conference on Learning
  Theory}, COLT '19, pages 1853--1902, 2019.

\bibitem[KR17]{KiferR17}
Daniel Kifer and Ryan~M. Rogers.
\newblock A new class of private chi-square tests.
\newblock In {\em Proceedings of the 20th International Conference on
  Artificial Intelligence and Statistics}, AISTATS '17, pages 991--1000. JMLR,
  Inc., 2017.

\bibitem[Kru83]{Kruskal83}
Clyde~P. Kruskal.
\newblock Searching, merging, and sorting in parallel computation.
\newblock {\em IEEE Transactions on Computers}, C-32(10):942--946, 1983.

\bibitem[KSF17]{KakizakiSF17}
Kazuya Kakizaki, Jun Sakuma, and Kazuto Fukuchi.
\newblock Differentially private chi-squared test by unit circle mechanism.
\newblock In {\em Proceedings of the 34th International Conference on Machine
  Learning}, ICML '17, pages 1761--1770. JMLR, Inc., 2017.

\bibitem[KU20]{KamathU20}
Gautam Kamath and Jonathan Ullman.
\newblock A primer on private statistics.
\newblock {\em arXiv preprint arXiv:2005.00010}, 2020.

\bibitem[Lei84]{Leighton84}
Tom Leighton.
\newblock Tight bounds on the complexity of parallel sorting.
\newblock In {\em Proceedings of the 16th Annual ACM Symposium on the Theory of
  Computing}, STOC '84, pages 71--80, New York, NY, USA, 1984. ACM.

\bibitem[MS08]{MahalanabisS08}
Satyaki Mahalanabis and Daniel Stefankovic.
\newblock Density estimation in linear time.
\newblock In {\em Proceedings of the 21st Annual Conference on Learning
  Theory}, COLT '08, pages 503--512, 2008.

\bibitem[MT07]{McSherryT07}
Frank McSherry and Kunal Talwar.
\newblock Mechanism design via differential privacy.
\newblock In {\em Proceedings of the 48th Annual IEEE Symposium on Foundations
  of Computer Science}, FOCS '07, pages 94--103, Washington, DC, USA, 2007.
  IEEE Computer Society.

\bibitem[NP33]{NeymanP33}
Jerzy Neyman and Egon~Sharpe Pearson.
\newblock Ix. on the problem of the most efficient tests of statistical
  hypotheses.
\newblock {\em Philosophical Transactions of the Royal Society of London.
  Series A, Containing Papers of a Mathematical or Physical Character},
  231(694-706):289--337, 1933.

\bibitem[Pip87]{Pippenger87}
Nicholas Pippenger.
\newblock Sorting and selecting in rounds.
\newblock {\em SIAM Journal on Computing}, 16(6):1032--1038, 1987.

\bibitem[SGHG{\etalchar{+}}19]{SwanbergGGRGB19}
Marika Swanberg, Ira Globus-Harris, Iris Griffith, Anna Ritz, Adam Groce, and
  Andrew Bray.
\newblock Improved differentially private analysis of variance.
\newblock {\em Proceedings on Privacy Enhancing Technologies}, 2019(3), 2019.

\bibitem[She18]{Sheffet18}
Or~Sheffet.
\newblock Locally private hypothesis testing.
\newblock In {\em Proceedings of the 35th International Conference on Machine
  Learning}, ICML '18, pages 4605--4614. JMLR, Inc., 2018.

\bibitem[SOAJ14]{SureshOAJ14}
Ananda~Theertha Suresh, Alon Orlitsky, Jayadev Acharya, and Ashkan Jafarpour.
\newblock Near-optimal-sample estimators for spherical {G}aussian mixtures.
\newblock In {\em Advances in Neural Information Processing Systems 27}, NIPS
  '14, pages 1395--1403. Curran Associates, Inc., 2014.

\bibitem[TZZ19]{TaoZZ19}
Chao Tao, Qin Zhang, and Yuan Zhou.
\newblock Collaborative learning with limited interaction: Tight bounds for
  distributed exploration in multi-armed bandits.
\newblock In {\em Proceedings of the 60th Annual IEEE Symposium on Foundations
  of Computer Science}, FOCS '19, pages 126--146, Washington, DC, USA, 2019.
  IEEE Computer Society.

\bibitem[Ull18]{Ullman18}
Jonathan Ullman.
\newblock Tight lower bounds for locally differentially private selection.
\newblock {\em arXiv preprint arXiv:1802.02638}, 2018.

\bibitem[USF13]{UhlerSF13}
Caroline Uhler, Aleksandra Slavkovi{\'c}, and Stephen~E. Fienberg.
\newblock Privacy-preserving data sharing for genome-wide association studies.
\newblock {\em The Journal of Privacy and Confidentiality}, 5(1):137--166,
  2013.

\bibitem[Val75]{Valiant75}
Leslie~G. Valiant.
\newblock Parallelism in comparison problems.
\newblock {\em SIAM Journal on Computing}, 4(3):348--355, 1975.

\bibitem[VS09]{VuS09}
Duy Vu and Aleksandra Slavkovi{\'c}.
\newblock Differential privacy for clinical trial data: Preliminary
  evaluations.
\newblock In {\em 2009 IEEE International Conference on Data Mining Workshops},
  ICDMW '09, pages 138--143. IEEE, 2009.

\bibitem[War65]{Warner65}
Stanley~L. Warner.
\newblock Randomized response: A survey technique for eliminating evasive
  answer bias.
\newblock {\em Journal of the American Statistical Association},
  60(309):63--69, 1965.

\bibitem[WLK15]{WangLK15}
Yue Wang, Jaewoo Lee, and Daniel Kifer.
\newblock Revisiting differentially private hypothesis tests for categorical
  data.
\newblock {\em arXiv preprint arXiv:1511.03376}, 2015.

\bibitem[Yat85]{Yatracos85}
Yannis~G. Yatracos.
\newblock Rates of convergence of minimum distance estimators and
  {K}olmogorov's entropy.
\newblock {\em The Annals of Statistics}, 13(2):768--774, 1985.

\end{thebibliography}

\end{document}